\documentclass{aptpub1}

\usepackage{epsf}
\usepackage{upgreek}
\usepackage{wrapfig}
\usepackage{boxedminipage}
\usepackage{pifont}
\usepackage{euscript}
\usepackage{dcolumn}
\usepackage{float}
\usepackage{fancybox}
\usepackage{rotating}
\usepackage{multirow}
\usepackage{latexsym}
\usepackage{amssymb}
\usepackage{txfonts}
\usepackage{amsmath}
\usepackage{epsfig}
\usepackage{color}
\usepackage{verbatim}
\usepackage{subfigure}
\usepackage{mathrsfs}
\usepackage{url}
\usepackage{lscape}

\usepackage{graphicx}

\def\bfseries{\fontseries \bfdefault \selectfont \boldmath} %to get
\newtheorem{alg}{Algorithm}[section]

\newcommand{\ealg}{\end{alg}}
\newcommand{\balg}{\begin{alg}}

\newcommand{\bigO}{\EuScript{O}}

\newcommand{\bX}{\mathbf{X}}
\newcommand{\bx}{\mathbf{x}}

\newcommand{\by}{\mathbf{y}}

\newcommand{\bu}{\mathbf{u}}

\newcommand{\bY}{\mathbf{Y}}

\newcommand{\halmos}{\vspace{3mm} \hfill \mbox{$\Box$}}

\newcommand{\ben}{\begin{enumerate}}
\newcommand{\een}{\end{enumerate}}
\newcommand{\beq}{\begin{equation}}
\newcommand{\eeq}{\end{equation}}
\newcommand{\ei}{\end{itemize}}
\newcommand{\bex}{\begin{example}}
\newcommand{\eex}{\end{example}}
\newcommand{\berem}{\begin{remark}}
\newcommand{\erem}{\end{remark}}
\newcommand{\beprop}{\begin{proposition}}
\newcommand{\eprop}{\end{proposition}}

\newcommand{\Var}{\text{Var}}

\newcommand{\bv}{\mathbf{v}}

\renewcommand{\epsilon}{\varepsilon}
\renewcommand{\rho}{\varrho}
\renewcommand{\log}{\ln}
\renewcommand{\hat}{\widehat}
\renewcommand{\leq}{\leqslant}
\renewcommand{\geq}{\geqslant}

\newcommand{\argmax}{\mathop{\rm argmax}}
\newcommand{\argmin}{\mathop{\rm argmin}}

\newcommand{\Geo}{{\sf Geom}}

\newcommand{\Em}{\mathbb E}
\newcommand{\Pm}{\mathbb P}
\newcommand{\R}{\mathbb R}

\newcommand{\scF}{\mathscr{F}}

\newcommand{\gvn}{\,|\,}

\newcommand{\e}{\text{e}}

\newcommand{\di}{\text{d}}

\def\acro#1#2{\vskip4pt\hbox to\textwidth{\normalsize
\hbox to5pc{#1\hfill}\vtop{\advance\hsize by
-5pc\raggedright\noindent#2}}}

\def\symbol#1#2{\vskip4pt\hbox to\textwidth{\normalsize
\hbox to5pc{#1\hfill}\vtop{\advance\hsize by
-5pc\raggedright\noindent#2}}}

\newcommand{\iidsim}{\stackrel{\text{iid}}{\sim}}
\newcommand{\simiid}{\iidsim}

\newcommand{\chk}[1]{}

\newcommand{\idef}{\stackrel{\text{def}}{=}}

\usepackage{psfrag}
\usepackage{algorithm}
\usepackage[noend]{algorithmic}
\algsetup{indent=2em}

\renewcommand{\argmax}{\operatornamewithlimits{argmax}}

\newcommand{\Ftail}{\overline F}
\newcommand{\Fdtail}{\overline{F^{*d}}}
\newcommand{\Fdmtail}{\overline{F^{*(d-1)}}}
\newcommand{\Exp}{{\mathbb E}\,}
\newcommand{\Ind}{\mathbb I}

\newcommand{\nat}{{\mathbb N}}

\newcommand{\dd}{{\mathrm{d}}}

\newcommand{\II}{\mathbb{I}}

\newcommand{\simapprox}{\stackrel{\mathrm{approx}}{\sim}}
\newcommand{\sV}{\mathscr{V}}
\newcommand{\scG}{\mathscr{G}}

\newif\ifnotes\notestrue
\authornames{Botev, Ridder, Rojas-Nandayapa } % insert the authors here for use in running head
\shorttitle{Semiparametric CE method}

\begin{document}

\title{Semiparametric Cross Entropy  for rare-event simulation}

\authorone[The University of New South Wales]{Z. I. Botev}
\authortwo[Vrije Universiteit ]{A. Ridder}
\authorthree[The University of Queensland]{L. Rojas-Nandayapa}

\addressone{School of Mathematics and Statistics,
University of New South Wales,
Sydney, NSW 2052
Australia } % Your postal address goes here.

\addresstwo{School of Mathematics and Physics,
The University of Queensland,
Brisbane, QLD 4072, Australia}

\addressthree{Department Econometrics and Operations Research,
Vrije Universiteit, 1081 HV, Amsterdam }

\begin{abstract}
  The Cross Entropy method is a well-known adaptive importance sampling method for rare-event probability estimation, which requires estimating an optimal importance sampling density within a parametric class.   In this article we   estimate an optimal importance sampling density within a wider semiparametric class
	of distributions. We show that this semiparametric version of the Cross Entropy method
	frequently yields efficient estimators. We illustrate the excellent practical performance of the method with
	numerical experiments and show that for the problems we consider it typically outperforms alternative schemes by orders of magnitude.
\end{abstract}



\renewcommand{\thefootnote}{\arabic{footnote}}
\setcounter{footnote}{0}

\section{Introduction}
In this article we consider the problem of estimating rare-event
probabilities of the form
\[
\ell=\Pm(S(\bX)>\gamma), \quad \bX=(X_1,\ldots,X_d),
\]
where $S(\bx)=x_1+\cdots+x_d$ and $X_1,\ldots,X_d$ are (possibly dependent)
random variables.
We call these the jump variables.
Such estimation problems arise in various contexts,
see, for example, \cite{mcis:asmgly07,AKR05,EMB97}.
We describe an adaptive importance sampling algorithm, which can be
viewed as the semiparametric version of the well-known Cross Entropy (CE)
method for estimation of rare-event probabilities \cite{kroeseRub11}.
The main ingredients  of the semiparametric CE method are as follows.

First, similar to \cite{mcis:botev10,botev2013markov} we use a  Markov Chain Monte Carlo (MCMC) algorithm to obtain
random variables distributed according to the minimum variance importance sampling density.
In our context the minimum variance importance sampling density is simply the density of the vector $\bX$ conditioned on the rare event $S(\bX)>\gamma$. Second, with  the MCMC sample at hand, we construct
a conditional (or a Rao-Blackwell) estimator of each of the marginal densities of the minimum variance importance sampling density. Finally, we use the product of these (estimated) marginal densities as our importance sampling density in order to estimate $\ell$. Under idealized conditions that ignore the error arising from the MCMC sampling, we show that the resulting estimator achieves either logarithmic or bounded relative error efficiencies. The strength of the method is not only that it outperforms the currently recommended  estimation procedures for heavy-tailed probabilities, but that the exact same procedure is efficient in problems with light-tailed probabilities.  For example, we show that unlike any existing procedures, the method is efficient in the  Weibull case for all values of the tail index $\alpha$, even in the light-tailed case with $\alpha>1$.

Numerical experiments show that, despite the heuristic nature of
the MCMC step, the estimator can in practice be frequently more
reliable and efficient than tailor-made importance sampling schemes.
In other words, an advantage of the methodology advocated here is that a
single broadly-applicable heuristic algorithm provides satisfactory
practical performance on a range of different estimation problems (both in  light- and heavy-tailed cases) and
frequently this performance is superior to  estimation schemes that
are specifically designed  to a particular rare-event estimation problem.

The rest of the paper is organized as follows.
In Section~\ref{sec:CE} we quickly review the parametric CE method and
introduce its semiparametric version.
This is followed by a number of examples with details about the practical
implementation of the estimator. The examples aims to  demonstrate
the superior performance of the proposed algorithm compared to
existing estimation algorithms on a number of prototypical examples. In Section~\ref{sec:analysis} we provide theoretical analysis of the efficiency of a simple version of the estimator  for light- and
heavy- tailed random variables. Finally, Section~\ref{sec:conclusion} gives some concluding remarks.

\section{Cross Entropy method}
\label{sec:CE}
\subsection{Parametric Cross Entropy method}
In order to introduce the semiparametric version of the CE method,
we briefly review the CE method itself.
Let $f(\bx)$ be the joint density of the vector $\bX=(X_1,\ldots,X_d)$
and suppose that it is part of the parametric family
\begin{equation}\label{e:family_F}
\scF=\Big\{f(\cdot;\bv): \R^d\to \R_{\geq 0}: \int f(\bx;\bv)\,\di\bx =1;
\bv\in \sV\Big\},
\end{equation}
where $\sV\subset\R^p$ is the feasible parameter set.
The assumption is that $f(\bx)\equiv f(\bx;\bu)\in\scF$ for some $\bu\in \sV$.
Then, the objective is to find a parameter $\bv\in\sV$  that yields a
good  importance sampling estimator of the form:
\begin{equation}
\label{CE est}
\hat\ell_\mathrm{CE}=\frac{1}{m}\sum_{i=1}^m
\II\{S(\bY_i)>\gamma\} \frac{f(\bY_i;\bu)}{f(\bY_i;\bv)},
\qquad \bY_1,\ldots,\bY_m \simiid f(\by;\bv)\;.
\end{equation}
In the CE method the best parameter $\bv^*\in\sV$ is the one which minimizes
the cross entropy distance between $f(\cdot;\bv)\in\scF$ and the
zero-variance importance sampling density
\[
\pi(\bx)=\frac{\II\{S(\bx)>\gamma\}f(\bx)}{\Pm(S(\bX)>\gamma)}\;.
\]
In other words,
\begin{equation}
\label{argmax for v*}
 \bv^*=\argmin_{\bv\in\sV} \int\pi(\bx)\,
 \log\left(\frac{\pi(\bx)}{f(\bx;\bv)}\right)\di\bx
 =\argmax_{\bv\in\sV}  \int\pi(\bx)\,\log f(\bx;\bv) \,\di\bx\;.
\end{equation}
In practice the integral
$\int\pi(\bx)\log\left(\frac{\pi(\bx)}{f(\bx;\bv)}\right)\di \bx$
is estimated   from a preliminary simulation so that we obtain
the estimator of $\bv^*$:
\begin{equation}
\label{argmax for hatv}
\hat{\bv^*}=\argmax_{\bv\in\sV} \sum_{i=1}^n \log f(\bX_i,\bv),
\end{equation}
where $\bX_1,\ldots,\bX_n$ is an approximate sample from $\pi$ obtained
via Markov chain Monte Carlo (MCMC) sampling over the restricted set $\mathscr{S}_\gamma$, see \cite{mcis:chan} and  Remark~\ref{rem:Gibbs} below.
In this way we  use MCMC to learn about the optimal (in cross entropy sense)
parameter $\bv^*$. In many applications the parametric
density $f(\cdot;\bv)$ is of product form:
$f(\bx;\bv)=\prod_{i=1}^d f_i(x_i;v_i)$.
For the special case where each $f_i(x_i;v_i)$ belongs to a
one-parameter exponential family parametrized by the mean
\cite[Pages 69-70]{varred:cebook},
the solution of \eqref{argmax for hatv} is given by the
maximum-likelihood estimator of the mean vector:
\[
 \hat{v^*_i}=\frac{1}{n}\sum_{j=1}^n  X_{j,i},\quad i=1,\ldots,d\;,
\]
where $X_{j,i}$ is the $i$-th coordinate of the $j$-th sample
$\bX_j$.
We thus use the importance sampling estimator \eqref{CE est}
with $\bv=\hat\bv^*$.
\begin{remark}[Generating $\bX_1,\ldots,\bX_n$ via Gibbs sampling]
\label{rem:Gibbs}
 In our discussion we  assume that the conditional
densities $\pi(x_i\gvn\bx_{-i})$ are available in closed form. We can thus use the following Gibbs sampling procedure to obtain $\bX_1,\ldots,\bX_n\simapprox \pi$.
\begin{alg} [Gibbs Sampler]~ \label{alg.gibbs}
\rm\begin{algorithmic}
\REQUIRE An initial state $\bX_0\sim f(\bx)$ and sample size $n$.
\FOR{$t=0,\ldots,n-1$}
 \STATE{Set $\bY=\bX_{t}$.
}
 \FOR{$i=1,\ldots,d$}
   \STATE{Draw $Y_i\sim\pi(y_i \gvn Y_1,\ldots,Y_{i-1},X_{t,i+1},\ldots,X_{t,d})$. }
\ENDFOR
\STATE{Set $\bX_{t+1} = \bY$.}
\ENDFOR
\end{algorithmic}
\end{alg}

\end{remark}

\subsection{Semiparametric Importance sampling}
\label{subsec:semiparm_IS}
Recall that the original CE method aims to find the best importance
sampling density
$f(\cdot;\bv^*)\in\scF$ within the parametric family
\eqref{e:family_F}; namely
by solving the parametric optimization program \eqref{argmax for v*}.
In contrast, in the semiparametric CE method the objective is to find
the optimal importance sampling density amongst
a family of densities given by some common property.
Again, the optimality criterion is to minimize the cross-entropy
distance from the the zero-variance density.
Denote by $\scG_1$ the set of all single-variate probability
density functions; that is, $g(x):\R\to\R_{\geq 0}$ is absolute continuous
with $\int g(x)\,\di x=1$.
Let $\scG$ be the family of product-form densities on $\R^d$:
\[
\scG=\Big\{g(\cdot): \R^d\to \R_{\geq 0}:
g(\bx)=\prod_{i=1}^d g_i(x_i); g_i\in\scG_1, i=1,\ldots,d\Big\}.
\]
In this paper we consider $\scG$ as the target set of importance sampling
densities. Hence, the objective is to
solve the functional optimization program
$
\min_{g\in\scG}\int\pi(\bx)\,
\log\left(\frac{\pi(\bx)}{g(\bx)}\right)\di\bx.
$
This is equivalent to
\begin{equation}\label{e:semiparmprogram}
g(\bx)=\argmin_{g_1,\ldots,g_d\in\scG_1} \int\pi(\bx)\,
\log\left(\frac{\pi(\bx)}{\prod_{i=1}^d g_i(x_i)}\right)\di\bx
=\argmax_{g_1,\ldots,g_d\in\scG_1} \int\pi(\bx)\,
\log\left(\prod_{i=1}^d g_i(x_i)\right)\di\bx.
\end{equation}

\begin{lemma}\label{Lemma1}
Let $\pi_i(x_i)$ be the $i$-th marginal of the zero-variance density $\pi(\bx)$.
Then the solution to the semiparametric CE program \eqref{e:semiparmprogram}
is $g_i=\pi_i$ for all $i=1,\ldots,d$.
In other words, the optimal importance sampling density within the space
of all product-form densities is the one given by the product of the
marginals of $\pi(\bx)$.
\end{lemma}
The proof is given in the Appendix. 
In practice the marginal densities of $\pi$ are not available
(just like the exact $\bv^*$ in \eqref{argmax for v*} is not available) and
need to be estimated from simulation. Here we use the estimators
\begin{equation}
\label{hat g}
 \hat\pi_i(y_i)=\frac{1}{n}\sum_{k=1}^n\pi(y_i\gvn\bX_{k,-i}),\quad i=1,\ldots,d\;,
\end{equation}
where
\begin{itemize}
\item   $\bX_1,\ldots,\bX_n$ is an approximate sample from $\pi$ obtained
via  Gibbs sampling as in \eqref{argmax for hatv} (see also Remark~\ref{rem:Gibbs});
\item the vector $\bX_{k,-i}$ is the
same as $\bX_k$ except that the $i$-th component is removed;
\item 
$\pi(x_i\gvn\bX_{k,-i})$ is the conditional density of $x_i$ given all the
other components of $\bX_k$.
\end{itemize}
 The estimator \eqref{hat g} is motivated by the simple
identity:
\begin{align*}
\Em_\pi& [\hat\pi_i(y)]=
\frac{1}{n}\sum_{k=1}^n\Em_\pi[\pi(y\gvn\bX_{k,-i})]
=\Em_\pi[\pi(y\gvn\bX_{-i})]\\
&=\Em_\pi[\pi(y\gvn X_1,\ldots,X_{i-1},X_{i+1},\ldots,X_d)]\\
&=\int\, \pi(y\gvn x_1,\ldots,x_{i-1},x_{i+1},\ldots,x_d)\,\pi(\bx)\;\di\bx\\
&=
\int\, \frac{\pi(x_1,\ldots,x_{i-1},y,x_{i+1},\ldots,x_d)}{\pi(x_1,\ldots,x_{i-1},x_{i+1},\ldots,x_d)}\,\pi(\bx)\;\di\bx
\\
&=
\int \frac{\pi(x_1,\ldots,x_{i-1},y,x_{i+1},\ldots,x_d)}{\pi(x_1,\ldots,x_{i-1},x_{i+1},\ldots,x_d)}\;\di \bx_{-i}\;\times\;
\overbrace{\int \pi(x_1,\ldots,x_d)\,\di x_i}^{\pi(x_1,\ldots,x_{i-1},x_{i+1},\ldots,x_d)}
\\
&=\int\pi(x_1,\ldots,x_{i-1},y,x_{i+1},\ldots,x_d)\;\di \bx_{-i}=\pi_i(y)\;.
\end{align*}
\noindent
We define the approximation to the optimal semiparametric CE solution by
the product of marginal density estimators \eqref{hat g}, that is,
\begin{equation}\label{e:appoptsemi}
\hat g(\by)\idef \prod_{i=1}^d\hat\pi_i(y_i).
\end{equation}

Then we estimate $\ell$ by the importance sampling estimator
\begin{equation}
\label{semiparam IS}
\hat\ell=\frac{1}{m}\sum_{i=1}^m  \II\{S(\bY_i)>\gamma\}\frac{f(\bY_i)}{\hat g(\bY_i)},\qquad
\bY_1,\ldots,\bY_m\simiid \hat g(\by)\;,
\end{equation}

Note that, conditional on $\bX_1,\ldots,\bX_n$, each $\hat\pi_i$ is an equally weighted mixture of $n$ densities (with $k$-th component $\pi(y_i\gvn \bX_{k,-i})$) and hence sampling
$Y_i\sim\hat\pi_i(y_i)$ can be performed using the composition method \cite{mcis:handbook}[Page 53].
In other words, choose a component of the mixture at random  by generating $K$ uniformly from the set
of integers $\{1,\ldots,n\}$. Then, given  $K=k$, sample $Y_i$ from the $k$-th mixture component $Y\sim\pi(y_i\gvn \bX_{k,-i})$. Finally,  deliver $Y_i$ as a realization from  $\hat\pi(y_i)$ and $(Y_1,\ldots,Y_d)$ as a realization from $\hat g(\by)$.
\begin{remark}[Using exact conditional density]
\label{remark}
Note that once we have sampled $Y_1,\ldots,Y_{d-1}$ from $\hat\pi_1,\ldots,\hat\pi_{d-1}$, respectively, we have the option of sampling the final $Y_d$ from the exact conditional $\pi(y_d\gvn Y_1,\ldots,Y_{d-1})$, instead of from the $d$-th marginal $\hat\pi_d$. This reduces the cross entropy distance to $\pi$ even further and yields the alternative and typically more efficient
estimator \eqref{semiparam IS} with $\hat g(\by)$ redefined as
\[
\hat g(\by)\leftarrow \hat\pi_1(y_1)\times\cdots \times\hat\pi_{d-1}(y_{d-1})\times \pi(y_d\gvn y_1,\ldots,y_{d-1})\;.
\]
\end{remark}

\section{Examples and Practical Implementation}
In this section we consider the prototypical problem of estimating $\Pm(X_1+\cdots+X_d>\gamma)$, where the jumps
$X_1,X_2,\ldots$ may or may not be dependent. 
%In the previous section we have shown that the importance sampling estimator \eqref{} can be efficient when estimating $\Pm(S(\bX)>\gamma)$. This does not establish that the  estimator \eqref{} is efficient, because in practice we use the MCMC approximation $\hat g$.  Despite the heuristic nature of \eqref{},
In the case of independent jumps, the proposed importance sampling can yield  practical performance surpassing that of well established alternative estimation procedures such as the Asmussen-Kroese (AK) estimator \cite{asmkro06,asmussen2012error}. This is in part due to the fact that our estimator incorporates the ingenious exchangeability and conditioning  proposed in \cite{asmkro06}.   First, recall that the AK  estimator in \cite{asmkro06} based on one replication is given by
\[
\hat\ell_\mathrm{AK}=d\Ftail\Big(\Big(\gamma-\sum_{j=1}^{d-1}X_j\Big)\vee\max_{j<d} X_j\Big),\quad X_1,\ldots,X_{d-1}\simiid F\;.
\]
 The motivation for the estimator is the identity
$
\ell=d \,\Pm\Big(X_1+\cdots+X_d>\gamma, X_d=M_d\Big)=d\,\Em\Ftail\Big(\Big(\gamma-\sum_{j=1}^{d-1}X_j\Big)\vee\max_{j<d} X_j\Big)\;,
$
where $x\vee y=\max\{x,y\}$ and $M_d\idef \max_{j\leq d} X_j$.
This conditional estimator enjoys excellent practical performance
%and is the current  state-of-the-art estimation procedure
for the problems we consider below. For further details we refer to \cite{asmussen2012error,hartinger2009efficiency}, where the authors  prove that the estimator is a vanishing relative error one.

We obtain an estimator that  outperforms $\hat\ell_\mathrm{AK}$  in terms of (estimated) relative time variance by exploiting  the decomposition proposed in \cite{juneja2007estimating} and the ex
\[
\begin{split}
\ell&=\Pm(M_d>\gamma)+\Pm(S(\bX)>\gamma,M_d<\gamma)\\
&= \Pm(M_d>\gamma)+d\,\Pm(S(\bX)>\gamma,X_d=M_d<\gamma),\quad  \textrm{ by exchangeability of jumps}\\
&=1-\Pm(M_d<\gamma)+ d\;\Pm(X_d=M_d<\gamma) \;\Pm(S(\bX)>\gamma\gvn X_d=M_d<\gamma)\\
&=\overbrace{1-[F(\gamma)]^d}^{\textrm{dominant term}}+ \Pm(M_d<\gamma)\overbrace{\widetilde\Pm\Big(S(\bX)>\gamma\Big)}^{\textrm{residual probability}}\;,
\end{split}
\]
where the new probability measure $\widetilde\Pm(\cdot)=\Pm(\cdot\gvn X_d=M_d<\gamma)$ with corresponding density 
\[
\widetilde{f}(\bx) = f(\bx\gvn X_d=M_d<\gamma)
=\frac{d\,f(\bx)}{[F(\gamma)]^d}\,\II\left\{M_d<\gamma, X_d=M_d\right\}. 
\]
Estimating the residual probability, we obtain the one replication estimator
for $\ell$ as
\begin{equation}
\begin{split}
\label{mod est}
\hat\ell=&1-[F(\gamma)]^d+ \frac{\widetilde f(\bY)}{\hat g(\bY) }\II\Big\{S(\bY)>\gamma\Big\},\qquad \bY\sim \hat g(\by)\;,
%\frac{\Ftail\Big(\Big(\gamma-\sum_{j=1}^{d-1}Y_j\Big)\vee\max_{j<d} Y_j\Big)-\Ftail(\gamma)}{\prod_{j=1}^{d-1} \hat\pi_j(Y_j)}\prod_{j=1}^{d-1}f(Y_j),
\end{split}
\end{equation}

where $\hat g(\by)\idef \hat\pi_1(y_1)\cdots \hat\pi_{d-1}(y_{d-1})\;\pi(y_d\gvn y_1,\ldots,y_{d-1})$ is the estimated importance sampling pdf  described in Remark~\ref{remark}. 
%with target $\pi(\by)\varpropto \widetilde f(\by)\;\II\{S(\by)>\gamma\}$.  
%The  target which $\hat g$ attempts to estimate here is  the minimum variance density for the estimation of the residual probability:
%\[
%\pi(\bx)=\frac{f(\bx)\;\II\{\bx\in\mathscr{R}_\gamma\}}{\Pm(\bX\in\mathscr{R}_\gamma)}\;.
%\]
%This idea of splitting the tail probability into an asymptotically dominant term and a residual probability to achieve variance reduction was first proposed in \cite{juneja2007estimating}.
%The numerical evidence shows  that by itself this idea does not yield an efficient estimator, unless we apply the importance sampling scheme described here.

In the following examples we used the relative time variance product (RTVP) and the ratio of relative errors as a measure  of efficiency:
\[
\mathrm{Ratio}\idef\frac{\hat \sigma_\mathrm{AK}/\hat\ell_\mathrm{AK}}{\hat\sigma/\hat\ell}\;,\qquad \mathrm{RTVP}\idef \mathrm{Ratio}^2\times \frac{\tau_\mathrm{AK}}{\tau}\;,
\]
where $ \hat\sigma_\mathrm{AK}$ and $ \hat\sigma$  are the sample standard deviations of  $\hat\ell_\mathrm{AK}$ and $\hat\ell$ (all based on $m$
replications), respectively , and $\tau_\mathrm{AK}$ and $\tau$ are the CPU times taken to compute the respective estimators. The quantity $\tau$ includes the CPU time needed for the preliminary  MCMC simulations. 
\begin{example}[Weibull case]
\label{ex:weibull}
Here we wish to estimate $\Pm(X_1+\cdots+X_d>\gamma)$ and assume that each of the jumps $X_i$ has density $\alpha x^{\alpha-1}\e^{-x^\alpha}$
for $x>0$ and $0<\alpha<1$. Hence, $\Ftail(x)=\e^{-x^\alpha}$.
In comprehensive simulations studies the proposed estimator outperformed the Asmussen-Kroese (AK) estimator in terms of relative time variance for all values of the parameters $\alpha$ and $\gamma$. The improvement, however, was not uniform, see Table~\ref{tab: weibull}, where, for example for $\alpha=0.1$, we can see savings from as little as  $71$ times to as large as approximately $6000$. The general trend is for large gains for smaller  $\gamma$
and $\alpha>0.6$ or $\alpha<0.3$.  %Surprisingly, despite the vanishing relative error properties of the AK estimator for $\alpha<0.36$ (see \cite{asmussen2012error})
 The AK estimator was strongest in the range $\alpha\in [0.3,0.6]$
with values for $\alpha\not\in[0.3,0.6]$ rendering it  less efficient compared to \eqref{mod est}.
% For example, for $\alpha=0.1,\gamma=10^{15},d=10$ we obtained the estimate $ 1.8467/10^{13}$ with
%relative error $ 1.2/10^{10}$, which is roughly 93 times smaller than the corresponding AK relative error and makes estimator \eqref{mod est} approximately  3656 times more efficient.

%Due to space constraints, here we give only a small simulation study.
Note that the AK estimator is much faster to evaluate than \eqref{mod est}, but this speed is insufficient to offset the substantial gains in squared relative error (given by Ratio column).  
\renewcommand{\time}{\mbox{\tiny\textrm{x}}}
\begin{table}[H]
\caption{Comparison of importance sampling method with the AK estimator. Algorithmic parameters were chosen to be $n=10^3,m=10^6,d=10$. The AK estimator is based on $m=10^6$ replications.}
\label{tab: weibull}

\begin{center}
\begin{minipage}{0.45\linewidth}
{\footnotesize
\begin{tabular}{c|c|c|c|c}
\multicolumn{5}{c}{$\alpha=0.1$}  \\
 \hline
$\gamma$&  $\hat\ell$ & Rel. Err. &  Ratio & RTVP  \\
 \hline
             $10^{10}$ & $   4.54/ 10^{4}$&  $ 1.7/10^6$  &  $13^2$        &  71 \\
						$10^{11}$ &  $ 3.40/10^5$  &   $4.1/10^7$  &$22^2$   &    197 \\
 						$10^{12}$ &  $ 1.30/10^6$  &   $6.4/10^8$  &$72^2$   &    2071 \\
						$10^{13}$ & $2.16/10^8$ & $ 8/10^9$ &  $59^2$    &   1429  \\
            $10^{15}$ & $1.84/10^{13}$ & $1.3/10^{10}$ &    $125^2$    &     5944

\end{tabular}
}
\end{minipage}
\qquad
\begin{minipage}{0.45\linewidth}
{\footnotesize
\begin{tabular}{c|c|c|c|c}
\multicolumn{5}{c}{$\alpha=0.2$}  \\
 \hline
$\gamma$&  $\hat\ell$ & Rel. Err. &  Ratio & RTVP  \\
 \hline
             $10^4$ & $  1.97/ 10^{2}$&  $ 6.5/10^5$  & $3^2$ & 3.7\\
						$10^5$ & $  4.64/ 10^{4}$&  $ 1.8/10^5$  & $5.6^2$ & 12\\
					  $10^6$ &  $ 1.31/10^6$ &  $ 3/10^6$  &    $9.2^2$  &   33\\
   $10^7$ &  $ 1.23/10^{10}$  & $4.3/10^7$  &  $10^2$   &    42\\
		$10^8$ & $5.13/10^{17}$ &  $6.5/10^8$ &   $7^2$   &   $20$
\end{tabular}
}
\end{minipage}
%\qquad
%\begin{minipage}{0.45\linewidth}
%{\footnotesize
%\begin{tabular}{c|c|c|c|c}
%\multicolumn{5}{c}{$\alpha=0.4$}  \\
 %\hline
%$\gamma$&  $\hat\ell$ & Rel. Err. &  Ratio & RTVP  \\
 %\hline
             %$500$ & $   8.83/ 10^{5}$&  $ 9/10^4$  & $6.5^2$ & 15\\
			 %$10^3$ & $1.63/10^6$ &  $7.2/10^5$ &$7^2$    &    17 \\
					%$5000$ &	  $8.50/10^{13}$&  $2.8/10^5$ &  $3.3^2$ &      4   \\
					%$10^4$ & $5.39/10^{17}$ & $1.5/10^{5}$  &  $3.7^2$  &      5\\
					%$10^5$ & $3.76/10^{43}$ &  $3.7/10^6$ &  $3.4^2$  &     4
%
%\end{tabular}
%}
%\end{minipage}
\begin{minipage}{0.45\linewidth}

{\footnotesize
\begin{tabular}{c|c|c|c|c}
\multicolumn{5}{c}{$\alpha=0.6$}  \\
 \hline
$\gamma$&  $\hat\ell$ & Rel. Err. &  Ratio & RTVP  \\
 \hline
             $10^2$ & $  9.47/ 10^{6}$&  $  2.6/10^4$  & $19^2$ & 130\\
						$150$  & $7.83/10^8$ & $1.5/10^4$ & $41^2$  &   550  \\
						$200$   & $ 1.34/10^9$ & $1.5/10^4$ &  $63^2$    &  1376\\
						 $500$ & $1.83/10^{17}$ &  $ 1.7/10^4$ &  $5.5^2$  &     11\\
 						 $10^3$ & $7.00/10^{27}$ &  $9.5/10^5$ &    $6^2$  &    13
						
\end{tabular}
}
\end{minipage}
\qquad
\begin{minipage}{0.45\linewidth}
{\footnotesize
\begin{tabular}{c|c|c|c|c}
\multicolumn{5}{c}{$\alpha=0.9$}  \\
 \hline
$\gamma$&  $\hat\ell$ & Rel. Err. &  Ratio & RTVP  \\
 \hline
             $30$ & $  1.33/ 10^{4}$&  $ 9/10^4$  & $13^2$ & 50\\
 $40$ & $ 6.27/ 10^{7}$&  $ 9/10^4$  &  $78^2$   &   1758.7 \\
 $50$ & $ 2.25/ 10^{9}$ &  $1/10^3$ & $254^2$    &    17746 \\
 $60$ &  $7.01/10^{12}$ &   $1/10^3$   &$556^2$  &      87103 \\
 $100$&  $4.34/10^{22}$ &  $1/10^3$ &  $ 300^2$  &      23768
\end{tabular}
}

\end{minipage}

\end{center}
\end{table}

\begin{remark}[Efficient evaluation of $\hat g$]
If we define, $c_k\idef  \left(\gamma-\sum_{j\not=i} \bX_{k,j}\right)^+$, then   \eqref{hat g} simplifies to 
\[
\begin{split}
\hat\pi_i(y_i)=\frac{1}{n}\sum_{k=1}^n\pi(y_i\gvn \bX_{k,-i})&=\frac{1}{n}\alpha y_i^{\alpha-1}\e^{-y_i^\alpha}\sum_{k=1}^n  \II\{ y_i \geq c_k\}/\e^{-c_k}=f(x_i) \frac{1}{n}\sum_{k=1}^n  \II\{ y_i \geq c_{(k)}\}\times \e^{c_{(k)}} \;,
\end{split}
\]
where the term $\sum_{k=1}^n  \II\{ y_i \geq c_{(k)}\}\times \e^{c_{(k)}}$ can be  evaluated  for an arbitrary $y_i$ quickly by first computing and storing in memory the cumulative sums $\sum_{k=1}^i   \e^{c_{(k)}},\; i=1,\ldots,n$  and then using table look-up methods with $\bigO(n)$ time complexity.
\end{remark}

\end{example}

\begin{example}[Pareto case]
\label{ex:pareto}
Assume that the jumps $X_i$ have Pareto density and distribution functions given by
$
 f(x)=\alpha/x^{\alpha+1},\; F(x)=1-1/x^{\alpha},
	\; x\ge1.
$
The following table shows the results of a comparison with the AK estimator for different values of $\alpha$ and $\gamma$. Again,  the efficiency gains with the proposed method can be of the order of  $10^4$.
\begin{table}[H]
\caption{Comparison of importance sampling with the AK estimator for Pareto case. Here $n=10^3,m=10^6,d=10$.}
\label{tab: pareto}

\begin{center}
\begin{minipage}{0.45\linewidth}
{\footnotesize
\begin{tabular}{c|c|c|c|c}
\multicolumn{5}{c}{$\alpha=0.5$}  \\
 \hline
$\gamma-d$&  $\hat\ell$ & Rel. Err. &  Ratio & RTVP  \\
 \hline
 $10^{8}$ & $   1.00/ 10^{3}$&  $ 5.6/10^7$  &  $33^2$        &  209 \\
             $10^{10}$ & $   1.00/ 10^{4}$&  $ 5.8/10^8$  &  $107^2$        &  3007 \\
						$10^{11}$ &  $ 3.16/10^5$  &   $1.8/10^8$  &$176^2$   &    6270 \\
 						$10^{12}$ &  $ 9.99/10^6$  &   $5.92/10^9$  &$364^2$   &    34271 \\
						$10^{15}$ & $3.16/10^7$ & $ 1.9/10^{10}$ &  $584^2$    &   82494  
\end{tabular}
}
\end{minipage}
\qquad
\begin{minipage}{0.45\linewidth}
{\footnotesize
\begin{tabular}{c|c|c|c|c}
\multicolumn{5}{c}{$\alpha=1$}  \\
 \hline
$\gamma-d$&  $\hat\ell$ & Rel. Err. &  Ratio & RTVP  \\
 \hline
             $10^4$ & $  1.00/ 10^{3}$&  $  5.1/10^6$  & $7^2$ & 11\\
						$10^6$  & $1.00/10^5$ & $1.0/10^7$ & $38^2$  &   330  \\
						$10^8$   & $ 1.00/10^7$ & $1.4/10^9$ &  $91^2$    &  1711\\
						 $10^{10}$ & $1.00/10^{9}$ &  $ 2.61/10^{11}$ &  $42^2$  &     322\\
 						 $10^{13}$ & $1.00/10^{12}$ &  $3/10^{14}$ &    $24^2$  &    123
						
\end{tabular}
}
\end{minipage}
\begin{minipage}{0.45\linewidth}

{\footnotesize
\begin{tabular}{c|c|c|c|c}
\multicolumn{5}{c}{$\alpha=5$}  \\
 \hline
$\gamma-d$&  $\hat\ell$ & Rel. Err. &  Ratio & RTVP  \\
 \hline      
          $10^1$ & $  2.58/ 10^{4}$&  $ 1.5/10^4$  & $10^2$ & 66\\
             $10^2$ & $  1.06/ 10^{9}$&  $ 1.2/10^5$  & $4^2$ & 11\\
						$10^3$ & $  1.00/ 10^{14}$&  $ 1.13/10^6$  & $4^2$ & 11\\
					  $10^4$ &  $ 1.00/10^{19}$ &  $ 1/10^7$  &    $4.4^2$  &   11\\
   $10^5$ &  $ 1.00/10^{24}$  & $1.2/10^8$  &  $4^2$   &    11
\end{tabular}
}
\end{minipage}
\qquad
\begin{minipage}{0.45\linewidth}
{\footnotesize
\begin{tabular}{c|c|c|c|c}
\multicolumn{5}{c}{$\alpha=10$}  \\
 \hline
$\gamma-d$&  $\hat\ell$ & Rel. Err. &  Ratio & RTVP  \\
 \hline
 $5$ & $  1.75/ 10^{6}$&  $ 2.4/10^4$  & $30^2$ & 609\\
             $10$ & $  1.09/ 10^{9}$&  $ 9.93/10^5$  & $6^2$ & 22\\
 $10^2$ & $ 1.00/ 10^{19}$&  $ 8.8/10^6$  &  $4^2$   &   13 \\
 $500$ & $ 1.02/ 10^{26}$ &  $1.6/10^6$ & $5^2$    &    11 \\
 $1500$ &  $1.73/10^{31}$ &   $5.5/10^7$   &$4.4^2$  &      13 
\end{tabular}
}

\end{minipage}

\end{center}
\end{table}

\end{example}

%\begin{example}[Correlated Lognormal Insurance Claims]

%Compare with MAK estimator and AK

%\end{example}

\begin{example}[Compound Sum]
We are interested in estimating the tail probability of a compound sum of the form
$\Pm(X_1+\cdots+X_R>\gamma)$, where the jumps $X_i$ are iid with Weibull distribution with parameter $0<\alpha<1$, and (without loss of generality) $R\sim \Geo(\rho)$ is a geometric random variable with pdf $\rho(1-\rho)^{r-1},\;r=1,2,\ldots$.
We have $\Pm(S_R>\gamma)=\Pm(X_1+\cdots+X_R>\gamma)=$
\[
\begin{split}
\rho\sum_{r=1}^\infty (1-\rho)^{r-1}\Pm(S_r>\gamma)&=\rho\sum_{r=1}^\infty (1-\rho)^{r-1}\Pm(M_r>\gamma)+\rho\sum_{r=2}^\infty (1-\rho)^{r-1}\Pm(M_r<\gamma,S_r>\gamma)\\
&=\underbrace{\frac{\Ftail(\gamma)}{\Ftail(\gamma)+\rho F(\gamma)}}_{\textrm{dominant term}}+\frac{\rho(1-\rho) (F(\gamma))^2}{\Ftail(\gamma)+\rho F(\gamma)}\underbrace{
%\Em\Big[ R\;\Pm\Big(X_1=\max_{j\leq R}X_j<\gamma,S_R>\gamma\;\big|\; R\Big)\;\big|\; R\geq 2\Big],
\widetilde\Pm\Big(S_R>\gamma\Big)}_{\textrm{residual probability}}\;,
\end{split}
\]
where under the new probability measure $\widetilde \Pm$ we have $(R-1)\sim \Geo(\Ftail(\gamma)+\rho F(\gamma)) $ with pdf $\widetilde\Pm(R=r)=f_R(r),\; r=2,3,\ldots$  and $X_1,X_2,\ldots \simiid f(x)$ with pdf given by the truncated Weibull density $f(x)=\alpha x^{\alpha-1}\e^{-x^\alpha}/(1-\e^{-\gamma^\alpha}),\; 0<x<\gamma$. Hence, we can again  apply our importance sampling estimator  to estimate the residual probability $\widetilde\Pm(S_R>\gamma)$.
%\[
%\mathscr{R}_\gamma\equiv\{\bx: r=2,3,\ldots,\; \max_{j\leq r} x_j<\gamma,\; x_1+\cdots+x_r>\gamma\}\;.
%\]  
The minimum variance pdf for the estimation of the residual is
\[
\pi(\by,r)\varpropto   f_{R}(r)\prod_{j=1}^r f(y_j)\;\II\{S_r>\gamma\},
\]
which can be easily sampled from using the Gibbs sampler in Algorithm~\ref{alg.gibbs} by noting that
\[\
\pi(r\gvn \bY)\varpropto  f_{R}(r)\;\II\{r\geq r^*(\bY)\} ,\quad  r^*(\bY)\idef\min\{r: Y_1+\cdots+Y_r>\gamma\}\;.
\]
 %The estimator for the residual becomes
%\[
%\frac{f_{R}(R)}{\pi(R\gvn \bY)}\prod_{j=1}^\infty  \frac{f(Y_j)}{\hat\pi_j(Y_j)},
%\]
%where in practice we generate as many  $Y_j\sim \hat\pi_j$ as are necessary to compute $r^*(\bY)$.
%Note that the computation of the  likelihood ratio does not pose any difficulties, because  there is always a finite $r'$ such that $\hat\pi_j\equiv f$ for all $j\geq r'$ and  $\prod_{j=r'}^\infty \frac{f(Y_j)}{\pi_j(Y_j)}=1$. 

Table~\ref{tab: compound} gives the results of a number of numerical experiments.
The results of our proposed method are significantly better in all cases, except $\alpha=0.2$ with $1/\rho\in \{50,100\}$. In the latter case, the variance reduction achieved by  the proposed method is not sufficient to offset the computational cost of simulating compound sums of expected length of $1/\rho$. 
Note that for $\alpha\geq 0.5$, the proposed method can be thousands of times more efficient. Our proposed method is also more efficient than the recently proposed improved Asmussen-Kroese estimator \cite{ghamami2012improving}[Table 2].  For example, based on the reported variances and computing time in \cite{ghamami2012improving},  in terms of RTVP our  estimator is from
$8.5$ to $45$ times more efficient. We must note, however, that  the results given in Table 2 of
\cite{ghamami2012improving} appear to be  incorrect. For example, for $\rho=0.15,\alpha=0.75,\gamma=63.361$ Table 2  reports the estimate $5.23\times 10^{-4}$ with relative error of
$0.4\%$. In contrast, 
we obtained the estimate $5.38\times 10^{-4}$ with relative error $0.03\%$, which we verified with a  Crude Monte Carlo simulation using $10^9$ repetitions. 
%Thus, we are not certain whether  the estimated variances in Table 2 are reliable as well.  
 %This is possibly  due to the fact that we not only twist the distribution of each jump variable $X_i$, but also the distribution of $R$. In contrast,  the AK estimator does not change  the distribution of $R$.

\begin{table}[H]
\caption{Compound Weibull sum with expected number of jumps $1/\rho$. Here $n=10^4,m=10^6$.}
\label{tab: compound}

\begin{center}
\begin{minipage}{0.45\linewidth}
{\footnotesize
\begin{tabular}{c|c|c|c|c}
\multicolumn{5}{c}{$\alpha=0.2$ with $\gamma=10^6$ fixed}  \\
 \hline
$1/\rho$&  $\hat\ell$ & Rel. Err. &  Ratio & RTVP  \\
 \hline
 $5$ & $   6.56/ 10^{7}$&  $ 1.4/10^5$  &  $3.6^2$        &  9.6 \\
             $10$ & $   1.31/ 10^{6}$&  $ 3.1/10^5$  &  $2.8^2$        &  3.5 \\
						$20$ &  $ 2.65/10^6$  &   $5.1/10^5$  &$2.2^2$   &    1.2 \\
 						$50$ &  $ 6.81/10^6$  &   $1.7/10^4$  &$1.4^2$   &    0.03 \\
						$100$ & $1.42/10^5$ & $ 1.7/10^{4}$ &  $2^2$    &   0.04  
\end{tabular}
}
\end{minipage}
\qquad
\begin{minipage}{0.45\linewidth}
{\footnotesize
\begin{tabular}{c|c|c|c|c}
\multicolumn{5}{c}{$\alpha=0.5$ with $\gamma=500$ fixed}  \\
 \hline
$1/\rho$&  $\hat\ell$ & Rel. Err. &  Ratio & RTVP  \\
 \hline
		 $3$ & $7.34/10^{10}$ &  $7.3/10^{4}$ &    $4^2$  &   16\\
             $5$ & $  1.60/ 10^{9}$&  $  1/10^3$  & $4.1^2$ & 12\\
						$10$  & $1.17/10^8$ & $1.7/10^3$ & $ 47^2$  &    445  \\
						$20$   & $1.24 /10^5$ & $7.2/10^4$ &  $246^2$    &  7300\\
						 $50$ & $7.9/10^{3}$ &  $ 2.1/10^{4}$ &  $58^2$  &     110
						
\end{tabular}
}
\end{minipage}
\begin{minipage}{0.45\linewidth}

{\footnotesize
\begin{tabular}{c|c|c|c|c}
\multicolumn{5}{c}{$\alpha=0.8$ with  $\gamma=30/\rho$ depending on $\rho$}  \\
 \hline
$1/\rho$&  $\hat\ell$ & Rel. Err. &  Ratio & RTVP  \\
 \hline      
          $3$ & $  6.29/ 10^{11}$&  $ 1.2/10^3$  & $330^2$ & 46000\\
             $5$ & $  1.65/ 10^{11}$&  $ 6.4/10^4$  & $930^2$ & 200000\\
						$10$ & $  6.94/ 10^{12}$&  $ 3.8/10^4$  & $ 2561^2$ &  780000\\
					  $20$ &  $ 4.64/10^{12}$ &  $ 2.7/10^4$  &    $ 3636^2$  &   34000\\
            $50$ &  $ 3.68/10^{12}$  & $2.1/10^4$  &  $ 1485^2$   &     27000
\end{tabular}
}
\end{minipage}
\qquad
\begin{minipage}{0.45\linewidth}
{\footnotesize
\begin{tabular}{c|c|c|c|c}
\multicolumn{5}{c}{$\alpha=0.95$ with  $\gamma=30/\rho$ depending on $\rho$}  \\
 \hline
$1/\rho$&  $\hat\ell$ & Rel. Err. &  Ratio & RTVP  \\
 \hline
 $5$ & $  2.61/ 10^{13}$&  $ 4.8/10^4$  & $10^6$ & $>10^5$\\
 $10$ & $ 2.18/ 10^{13}$&  $ 3/10^4$  &  $>10^6$   &   $>10^5$ \\
 $20$ & $ 2.00/10^{13}$ &  $2.2/10^4$ & $>10^6$     &    40000 \\
 $50$ &  $1.91/10^{13}$ &   $1.9/10^4$   &$>10^6$   &   $>10^5$\\
 $100$ & $  1.88/ 10^{13}$&  $ 1.7/10^4$  & $>10^6$   &   $>10^5$
\end{tabular}
}

\end{minipage}

\end{center}
\end{table}

\end{example}

%\begin{example}[Improved Cross Entropy method]
%Use Portfolio loss example
%
%
%\end{example}
%
%
%
%
%\begin{example}[A Random walk hitting a non-convex set]
%
%\end{example}

\section{Robustness Properties of Semiparametric Cross Entropy Estimator}
\label{sec:analysis}
In this section we study the robustness properties of the estimator \eqref{semiparam IS} 
when  $\gamma\to\infty$ in some simplified  prototypical settings. Clearly, then $\ell=\ell(\gamma)=\Pm(S(\bX)>\gamma)\to 0$.
We are interested in the behavior of the standard error of the estimator in this regime,
specifically, relative to its mean $\ell$. Since we take a finite constant sample size,
it suffices to analyze the robustness of the single-run estimator of $\ell$:
\begin{equation}\label{est1}
Z = Z(\gamma) = \II\{S(\bX)>\gamma\}\frac{f(\bX)}{g(\bX)},
\end{equation}
where $\bX\sim g(\bx)=\prod_{i=1}^d g_i(x_i)=\prod_{i=1}^d\pi_i(x_i)$.
For our analysis we assume that the importance sampling density $g$ is available.
In practice we estimate $g$ via $\hat g$ from MCMC simulation
as we discussed in Section \ref{subsec:semiparm_IS}. In this respect, our analysis is similar in spirit 
to the one conducted for the parametric Cross Entropy method \cite{chan2011comparison}.
%This approach is similar to the one undertaken in the analysis of the
%parametric CE  \cite{mcis:chan11}, where the unknown optimal parameter $\bv^*$
%has to be estimated as well.
The estimator has bounded relative error if
$\limsup_{\gamma\to\infty} \sqrt{\Var(Z)}/\ell<\infty$, which is equivalent
to having bounded relative second moment \cite{mcis:L'Ecuyer10}:
\[
\limsup_{\gamma\to\infty}\frac{\Em Z^2}{\ell^2}<\infty.
\]

\begin{assumption}
In this section we assume that the jump variables $X_1,\ldots,X_d$
are positive continuous, and that they
are independent and identically distributed random variables
with right-unbounded support.
\end{assumption}

We denote by $F(x)$ the cdf of a jump $X_i$ with associated
pdf $f_1(x)$. Let $\Ftail(x)=1-F(x)$ be the tail cdf,
$F^{*d}$ be the $d$-fold convolution of $F$,
with $\Fdtail=1-F^{*d}$.
Note that the rare-event probability of interest is
$
\ell = \Pm(X_1+\cdots+X_d>\gamma) = \Fdtail(\gamma).
$
Furthermore, the $i$-th marginal $\pi_i$ of the zero-variance pdf can be rewritten as
\begin{align*}
\pi_i (x_i)& = \int_{\R^{d-1}_{>0}}\pi(\bx)\,\di x_1\cdots \di x_{i-1}\di x_{i+1}\cdots \di x_d\\
&= \int_{\R^{d-1}_{>0}}\frac{\II\{S(\bx)>\gamma\}\,f(\bx)}{\ell}\,
\di x_1\cdots \di x_{i-1}\di x_{i+1}\cdots \di x_d\\
&=\int_{\R^{d-1}_{>0}}\frac{\II\{S(\bx)>\gamma\}\,\prod_{j=1}^d f_1(x_j)}{\ell}\,
\di x_1\cdots \di x_{i-1}\di x_{i+1}\cdots \di x_d\\
&=\frac{f_1(x_i)}{\ell}\int_{\R^{d-1}_{>0}} \II\{x_1+\cdots+x_d>\gamma\}\,\prod_{j\neq i}f_1(x_j)\,
\di x_1\cdots \di x_{i-1}\di x_{i+1}\cdots \di x_d\\
&=\frac{f_1(x_i)}{\ell}\,\Pm(X_1+\cdots+X_{i-1}+X_{i+1}+\cdots+X_d>\gamma-x_i)=
\frac{f_1(x_i)\,\Fdmtail(\gamma-x_i)}{\Fdtail(\gamma)}\;.
\end{align*}

\noindent
Note that for $x_i>\gamma$ we clearly have $\Fdmtail(\gamma-x_i)=1$,
and thus $\pi_i(x_i)=f_1(x_i)/\ell$.
%In general, we  have for the relative error of $\hat\ell$
%\[
%\frac{m\Var(\hat\ell)}{\ell^2}=\Em (\pi(\bX)/g(\bX))^2=
%\int\frac{f^2(\bx)\II\{x_1+\cdots+x_d>\gamma\}}{\ell^2 \prod_{i=1}^d\pi_i(x_i)}\di \bx,
%\]
Hence, the single-run estimator $Z$ %of $\ell$ based on just one simulated value from $g$
can be written as
\begin{equation}
\label{est1x}
\begin{split}
Z &= \II\{S(\bX)>\gamma\}\frac{f(\bX)}{g(\bX)}
=\II\{S(\bX)>\gamma\}\,\prod_{i=1}^d\frac{f_1(X_i)}{\pi_i(X_i)}= \II\{S(\bX)>\gamma\}\,\prod_{i=1}^d \frac{\Fdtail(\gamma)}{\Fdmtail(\gamma-X_i)}
\end{split}
\end{equation}
%\[
%=\II\{X_1+\cdots+X_d>\gamma\}\,\exp\left[ \sum_i\Big(\log \Ftail^{(d)}(\gamma)-
%\log\Ftail^{(d-1)}(\gamma-X_i)\Big)\right].
%\]
Finally, using
$
\Em Z^2=\Em_g Z^2=\Em_g Zf(\bX)/g(\bX)=\Em_f Z,
$
we get for the second moment of estimator $Z$:
\begin{equation}
\label{e:EZ2}
%\Em_f\frac{\II\{X_1+\cdots+X_d>\gamma\}}{\prod_{i=1}^d \Ftail^{(d-1)}(\gamma-X_i)/\ell}=
\Em Z^2 = \Em_f\II\{S(\bX)>\gamma\}\,
\prod_{i=1}^d \frac{\Fdtail(\gamma)}{\Fdmtail(\gamma-X_i)}.
%\,\exp\left[ \sum_i\Big(\log \Ftail^{(d)}(\gamma)-
%\log\Ftail^{(d-1)}(\gamma-X_i)\Big)\right].
\end{equation}

\begin{proposition}
Suppose that the jumps $X_1,\ldots,X_d$ are i.i.d.\ with a light-tailed or 
a subexponential
Weibull or Pareto distribution. Then, the semiparametric importance sampling estimator \eqref{est1}
is at least logarithmically efficient as $\gamma\to\infty$. 
\end{proposition}
In the subsequent sections we prove this result by considering the heavy- and light-tailed cases separately. 

\subsection{Heavy-tailed case}
\label{HeavySection}
In this section we assume that all jumps $X_i$ are drawn from a subexponential distribution $F$
satisfying (for all integer $d$)
\begin{equation}
\lim_{\gamma\uparrow\infty}\frac{\Fdtail(\gamma)}{\Ftail(\gamma)}=d.
\end{equation}
%or alternatively  $\Ftail^{(d)}(\gamma)$ \adchange{$\sim$} $\Pm_f(\max_i X_i>\gamma)$ as $\gamma\uparrow\infty$.
In the sequel we shall frequently use the trivial property
\begin{equation}\label{e:trivial}
\Fdtail(x) \geq \Ftail(x),\qquad x\ge0.
\end{equation}
Furthermore, we shall need Kesten's bound Lemma 1.3.5(c) in \cite{EMB97},
which states that for every $\epsilon>0$ there exists a constant $c_1$ such that for all $d\geq 2$
\begin{equation}
\label{kesten}
\Fdtail(x)\leq c_1(1+\epsilon)^d \Ftail(x),\qquad x\ge0.
\end{equation}

Denoting the maximum  $M_d= \max_{i\leq d} X_i$,
%and the sum S=X_1+\cdots+X_d$,
we can decompose the relative second moment as follows:
\begin{equation}
\label{decomp}
\frac{\Em Z^2}{\ell^2} =
\frac{\Em \II\{ M_d>\gamma\}\,Z^2}{\ell^2}
+ \frac{\Em \II\{ M_d\leq \gamma\}\,Z^2}{\ell^2}.
\end{equation}
In Lemma \ref{lem:firstterm} we shall prove
%We now argue
that the first term is bounded as $\gamma\to\infty$.
Concerning the second term,
we examine its behavior for various common probability models
in the next two sections.

%
%%%%%%%%%%%%%%%% heavy tails first term %%%%%%%%%%%%%%%%%%%%

\begin{lemma}\label{lem:firstterm}
\[
\limsup_{\gamma\to\infty} \frac{\Em \II\{ M_d>\gamma\}\,Z^2}{\ell^2}<\infty.
\]
\end{lemma}

\begin{proof}
Since
$\II\{S(\bx)>\gamma\}\leq 1$,
%$M_d>\gamma-c$ implies that $S(\bX)>\gamma-c$,
we use \eqref{e:EZ2} to find
\begin{equation}\label{e:EZ2andMgg}
\Em \II\{ M_d>\gamma\}\,Z^2 \leq \Em_f\II\{ M_d>\gamma\}
\prod_{i=1}^d \frac{\Fdtail(\gamma)}{\Fdmtail(\gamma-X_i)}
\end{equation}
Then observe that, if $M_d>\gamma$, there exists at least one
jump $X_j>\gamma$, and, hence, that there is at least
one $j$ for which $\Fdmtail(\gamma-X_j)=1$.
For all other jumps it holds trivially
$\Fdmtail(\gamma-X_i)\geq \Fdmtail(\gamma)$, thus
it follows that \eqref{e:EZ2andMgg} is bounded from above by
\begin{align*}
\frac{\Em_f \II\{ M_d>\gamma\}\prod_{i=1}^d \Fdtail(\gamma)}%
{\prod_{i\neq j}^d \Fdmtail(\gamma-X_i)}
&\leq \Em_f \II\{ M_d>\gamma\} \;\frac{\prod_{i=1}^d \Fdtail(\gamma)}%
{\prod_{i\neq j}^d \Fdmtail(\gamma)}\\
&\quad =\Pm_f(M_d>\gamma) \; \frac{\big(\Fdtail(\gamma)\big)^d}{\big(\Fdmtail(\gamma)\big)^{d-1}}\\
&\quad\leq \frac{\big(\Fdtail(\gamma)\big)^{d+1}}{\big(\Fdmtail(\gamma)\big)^{d-1}}
=\big(\Fdtail(\gamma)\big)^2\;
\Big(\frac{\Fdtail(\gamma)}{\Fdmtail(\gamma)}\Big)^{d-1},
\end{align*}
where the last inequality follows from $\Pm_f(M_d>\gamma)\leq \Pm_f(S(\bX)>\gamma)=\Fdtail(\gamma)$.
Now we use the bounds \eqref{e:trivial} and \eqref{kesten} for
\[
\Big(\frac{\Fdtail(\gamma)}{\Fdmtail(\gamma)}\Big)^{d-1}
\leq \Big(\frac{\Fdtail(\gamma)}{\Ftail(\gamma)}\Big)^{d-1}
\leq c_1^{d-1}(1+\epsilon)^{d(d-1)}.
\]
Collecting all bounds we obtain
\begin{equation}\label{e:boundedfirstterm}
\begin{split}
& \frac{\Em \II\{ M_d>\gamma\}\,Z^2}{\ell^2}
= \frac{1}{\big(\Fdtail(\gamma)\big)^2} \, \Em_f\II\{ M_d>\gamma\}
\prod_{i=1}^d \frac{\Fdtail(\gamma)}{\Fdmtail(\gamma-X_i)}\\
& \quad \leq c_1^{d-1}(1+\epsilon)^{d(d-1)}<\infty.
\end{split}
\end{equation}
\halmos
\end{proof}
Since we have bounded relative error
for the first term in \eqref{decomp}, then we can at most have 
bounded relative error for estimator \eqref{est1}. 
For example, if the second term in \eqref{decomp}
 vanishes or is bounded, then \eqref{est1} has bounded relative error.
%In what follows we examine the behavior of  this term  for various common probability models.

%%%%%%%%%%%%%%%% Weibull second term %%%%%%%%%%%%%%%%%%%%
\subsubsection{Weibull distribution}
As in Example~\ref{ex:weibull}, 
here we assume that each of the jumps $X_i$ have density $\alpha x^{\alpha-1}\e^{-x^\alpha}$ for  $0<\alpha<1$.
The purpose is to analyze the second term in \eqref{decomp}.
\begin{lemma}\label{lem:secondtermweibull}
\[
\limsup_{\gamma\to\infty} \frac{\Em \II\{ M_d<\gamma\}\,Z^2}{\ell^2}=0.
\]
\end{lemma}

\begin{proof}
Denote $S_d=S(\bX)$.
Using \eqref{e:EZ2} and $\ell=\Fdtail(\gamma)$,
we get
\begin{align*}
& \frac{\Em \II\{ M_d<\gamma\}\,Z^2}{\ell^2}=
\Em_f  \II\{M_d< \gamma, S_d>\gamma\}\,\frac{\prod_{i=1}^{d-2} \Fdtail(\gamma)}%
{\prod_{i=1}^{d} \Fdmtail(\gamma-X_i)}.
\end{align*}
From the bounds \eqref{e:trivial} and \eqref{kesten},
we obtain that this expression can be bounded above by
\begin{align*}
\Em_f & \II\{M_d< \gamma, S_d>\gamma\}\,
\frac{\prod_{i=1}^{d-2} c_1(1+\epsilon)^d \,\Ftail(\gamma)}%
{\prod_{i=1}^{d} \Ftail(\gamma-X_i)}\\
&= c_2\Em_f  \II\{M_d<\gamma ,S_d>\gamma\}\,
\exp\Big(-(d-2)\gamma^\alpha+\sum_{i=1}^d(\gamma-X_i)^\alpha\Big).
\end{align*}
We now consider the following integral over the region
$\{\bx: 0<x_i<\gamma, \sum_i x_i>\gamma\}$:
\begin{align*}
& \Em_f  \II\{M_d<\gamma, S_d>\gamma\}\,
\exp\Big(-(d-2)\gamma^\alpha+\sum_{i=1}^d(\gamma-X_i)^\alpha\Big)\\
&=\alpha^d\idotsint \left(\prod_{i=1}^d x_i^{\alpha-1}\right)\,
\exp\Big(-(d-2)\gamma^\alpha +\sum_{i=1}^d\big((\gamma-x_i)^\alpha-x_i^\alpha\big)\Big)\di\bx
\end{align*}
After the change of variable $u_i=x_i/\gamma$ for all $i$  we obtain
that this integral is a Laplace-type integral:
\[
\alpha^d\gamma^{d\,\alpha}\underbrace{\idotsint_{\mathscr{D}}
h(\bu) \;\e^{-\gamma^\alpha \phi(\bu)}\,\di\bu}_{\textrm{Laplace-type}},
\]
\vspace{-0.5cm}
where:
\begin{align*}
\mathscr{D}&\idef \left\{\bu: 0< u_i< 1,\;\sum_i u_i> 1\right\}\\
h(\bu)&\idef  \prod_{i=1}^d u_i^{\alpha-1} \\
\phi(\bu)&\idef d-2+\sum_{i=1}^d \big(u_i^\alpha-(1-u_i)^\alpha\big)
\end{align*}
We now note the following properties of the Laplace integral.
First,  if $\bar{\mathscr{D}}$ denotes the closure of the open set $\mathscr{D}$,
the function $\phi(\bu)$ attains its unique global minimum within the bounded domain
$\bar{\mathscr{D}}\subseteq\mathbb{R}^d$ on the boundary
at $\bu^*=(1/d,\ldots,1/d)$. This can be seen either by applying the Lagrange constraint optimization method or more simply by noting that $u^\alpha-(1-u)^\alpha$  is monotonically increasing and $\phi(\bu)$
is a invariant to permutations of the components of $\bu$.
The minimum
\[
\phi(\bu^*)=d-2+d^{1-\alpha}-d^{1-\alpha}(d-1)^\alpha,
\]
%0<\alpha<1$
as a function of $d$ is such that for $d>2$ we have the strict
inequality $\phi(\bu)\geq \phi(\bu^*)>0$ for all $\bu\in\bar{\mathscr{D}}$,
see Figure~\ref{fig:laplace}. The point $\bu^*$ is not a critical point, because
$\frac{\partial \phi}{\partial u_i} (\bu)=
\alpha \left(u_i^{\alpha-1}+(1-u_i)^{\alpha-1} \right)>0$  for all $i$
and $\bu\in\mathscr{D}$.
\vspace{-.5cm}
\begin{figure}[H]
\begin{center}
\includegraphics[scale=0.5]{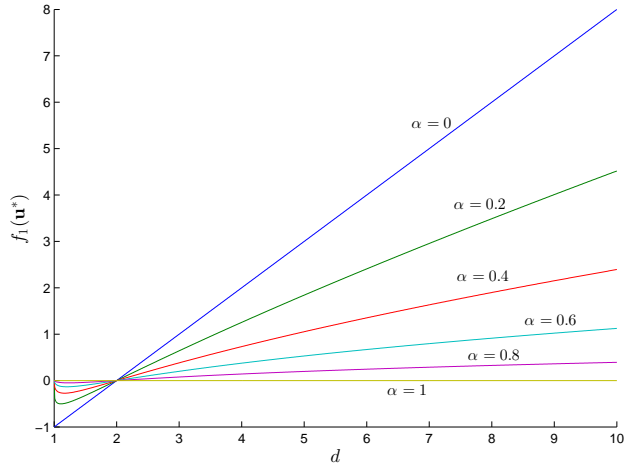}
\end{center}
\caption{The behavior of the function
$d-2+d^{1-\alpha}-d^{1-\alpha}(d-1)^\alpha$ for different values of the
parameter $\alpha$.}
\label{fig:laplace}
\end{figure}
\vspace{-.5cm}
Second, the function $h: \mathbb{R}^d\rightarrow \mathbb{R}$ is continuous and
the Hessian of the surface
$p(u_1,\ldots,u_{d-1})=\phi(u_1,u_2,\ldots,u_{d-1},1-u_1-u_2-\cdots-u_{d-1})$ is
\[
\frac{\partial^2 p}{\partial u_i \partial u_j}=\alpha(\alpha-1)\times
\begin{cases}
(1-\sum_{k<d}u_k)^{\alpha-2}- (\sum_{k<d}u_k)^{\alpha-2} & i\not=j\\
u_i^{\alpha-2}-(1-u_i)^{\alpha-2}+(1-\sum_{k<d}u_k)^{\alpha-2}-
(\sum_{k<d}u_k)^{\alpha-2} & i=j
\end{cases}\;,
\]
which when evaluated at $\bu^*$ yields the nondegenerate %invertible
Hessian matrix
\[
\alpha(\alpha-1)\left(d^{2-\alpha}-(1-1/d)^{\alpha-2}\right)\times
\begin{pmatrix}
2& 1& 1&\cdots& 1\\
1& 2& 1&\cdots& 1\\
1& 1& 2&\cdots& 1\\
\vdots& \vdots& \ddots&\ddots& \vdots\\
1& 1& 1&\cdots& 2
\end{pmatrix}\;.
\]
As a result of all these conditions we have the Laplace-type asymptotic
expansion \cite[Page 500]{wong2001asymptotic} at a boundary point,
which is not a critical point:
\[
\idotsint_{\mathscr{D}} h(\bu) \;\e^{-\gamma^\alpha \phi(\bu)}\di\bu =
\bigO(\gamma^{-\alpha(d+1)/2}\times \e^{-\gamma^\alpha \phi(\bu^*)} )\;,
\]
where the constant $\phi(\bu^*)>0$. It follows that
\begin{align*}
& \frac{\Em \II\{ M_d<\gamma\}\,Z^2}{\ell^2}
\leq c_2 \alpha^d \gamma^{d\alpha}
\idotsint_{\mathscr{D}} h(\bu) \;\e^{-\gamma^\alpha \phi(\bu)}\di\bu\\
&= \bigO\big(\gamma^{\alpha(d-1)/2}\times \e^{-\gamma^\alpha \phi(\bu^*)}\big)
= \bigO\big(\e^{\alpha(d-1)/2 \log\gamma-\gamma^\alpha \phi(\bu^*)}\big)
\to 0\quad (\gamma\to\infty).
\end{align*}
Hence the second term in \eqref{decomp} vanishes as $\gamma\to\infty$.
\halmos
\end{proof}

%%%%%%%%%%%%%%%%%% Pareto second term %%%%%%%%%%%%%%%%%%%%%
\subsubsection{Sum of Pareto random variables.}
As in Example~\ref{ex:pareto}, we assume that $X_i$'s are independent and Pareto distributed
random variables on $[1,\infty)$ with common parameter $\alpha>0$.
%\adchange{
The main result is the logarithmic efficiency of the second term
of \eqref{decomp}.
%}
%\adnote{
%Not yet BRE!
%}
\begin{proposition}\label{lem:secondtermpareto}
%\adchange{
For all $\epsilon>0$
%}
\begin{equation*}
\limsup_{\gamma\to\infty} \frac{\Em [Z^2;M_d\le \gamma]}{\ell^{2-\epsilon}}=0.
\end{equation*}
\end{proposition}
\par\bigskip\noindent
%\adchange{
\begin{proof}
The proof will be the result of a number of lemmas.
First, similarly as in Lemma \ref{lem:secondtermweibull} we utilize expression
\eqref{e:EZ2} for rewriting the second moment as a product, and then we apply
\eqref{e:trivial} and \eqref{kesten} to bound the factors. The result is
that it is enough to prove that
\begin{equation}\label{e:enough}
\limsup_{\gamma\to\infty} \frac{1}{\ell^{2-\epsilon}}\;
\Em_f\II\{M_d\le \gamma, S_d>\gamma\}\prod_{i=1}^d \frac{\Ftail(\gamma)}{\Ftail(\gamma-X_i)}=0.
\end{equation}
%}
Our approach is to consider a larger set containing $\{M_d\le \gamma,S_d>\gamma\}$.
For that purpose we define the we define the quantities
\begin{equation*}
H_n(\gamma):=\Em_f\bigg[\prod_{k=1}^n \frac{\Ftail(\gamma)}{\Ftail(\gamma-X_k)};B_n\bigg],\qquad n\ge2.
\end{equation*}
where
\begin{equation*}
%\adchange{
B_n=\{S_{n-1}\le \gamma, S_n>\gamma, M_n\le \gamma\},\qquad n=2,3,\ldots.
%}
\end{equation*}
Observe that $\{M_d\le \gamma,S_d>\gamma\}\subset\bigcup_{n=2}^d B_n$.
Further to this, observing that $\Ftail(\gamma)/\Ftail(\gamma-x)\le 1$
for all $x\ge 1$, we can set
\begin{equation*}
 \prod_{k=1}^d \frac{\Ftail(\gamma)}{\Ftail(\gamma-x_k)}\le\prod_{k=1}^n \frac{\Ftail(\gamma)}{\Ftail(\gamma-x_k)},
  \qquad n\le d.
\end{equation*}
In this way we arrive at the following inequality
%\adnote{I added the steps.}
\begin{equation}\label{e:sumhn}
\begin{split}
\Em_f\II\{&M_d\le \gamma, S_d>\gamma\}\prod_{i=1}^d \frac{\Ftail(\gamma)}{\Ftail(\gamma-X_i)}\\
&\le \sum_{n=2}^d \Em_f \bigg[\prod_{i=1}^d \frac{\Ftail(\gamma)}{\Ftail(\gamma-X_i)};B_n\bigg]\\
&\le \sum_{n=2}^d \Em_f \bigg[ \prod_{i=1}^n \frac{\Ftail(\gamma)}{\Ftail(\gamma-X_i)}
;B_n\bigg]
=\sum_{n=2}^dH_n(\gamma).
\end{split}
\end{equation}
Now, the quantities $H_n$ in the sum above can be written in integral form as
%\adnote{
%The upper limits are one larger because in my case
%$x_1+x_2+\cdots+x_{n-1}\leq \gamma$ and all $x_k\geq 1$ gives
%$x_1\leq \gamma -(x_2+\cdots+x_{n-1})\leq \gamma - (n-2)$.
%Etc.
%}
 \begin{align*}
  H_n(\gamma)
   &= \int_{B_n} \bigg(\prod_{k=1}^n \frac{\Ftail(\gamma)}{\Ftail(\gamma-x_k)}\bigg)
  \, \bigg(\prod_{k=1}^n f(x_k)\bigg)\ \dd x_n\,\dd x_{n-1}\dots\, \dd x_{2}\, \dd x_{1}\\
  &=\int\limits_{1}^{\gamma-(n-2)}\ \ \int\limits_{1}^{\gamma-x_1-(n-3)}\!\dots\!
  \int\limits_{1}^{\gamma-x_1-\dots-x_{n-2}}\!\!\!\!\!
   \int\limits_{(\gamma-x_1-\dots-x_{n-1})\vee 1}^{\gamma}
    \prod_{k=1}^n\bigg(\frac{\gamma-x_k}{\gamma}\bigg)^\alpha	
    \frac{\alpha}{x_k^{\alpha+1}}\ \dd x_n\,\dd x_{n-1}\dots\, \dd x_{2}\, \dd x_{1}.
 \end{align*}
%  \footnote{Leonardo: Notice the modification of the term $(\gamma-x_1-\dots-x_n)\vee1$
%   in the lower limit of the last integral.}
 Further, the change of variable $y_k=x_k/\gamma$ yields 
 \begin{equation*}
  \frac{\alpha^n}{\gamma^{n\alpha}}\!
  \int\limits_{\gamma^{-1}}^{1-(n-2)\gamma^{-1}}\ \ \int\limits_{\gamma^{-1}}^{1-y_1-(n-3)\gamma^{-1}}
   \!\dots\!\int\limits_{\gamma^{-1}}^{1-y_1-\dots-y_{n-2}}\!\!\!\!\!\!\int\limits_{(1-y_1-\dots-y_{n-1})\vee\gamma^{-1}}^{1}
    \prod_{k=1}^n\ L(y_k)\
     \dd y_n\,\dd y_{n-1}\dots\dd y_2 \, \dd y_1,
 \end{equation*}
 where
 \begin{equation}\label{e:Ly}
 L(y):=(1-y)^\alpha \;y^{-(\alpha+1)},\qquad y\in(0,1].
 \end{equation}
 In particular, it will be useful to write
\begin{equation}\label{e:defIn}
 H_n(\gamma)=\alpha^n\,\gamma^{-n\alpha}\,I_{n}(\gamma,1),
\end{equation}
  where the function $I_{n}(\gamma,1)$ is the multiple integral in the
  expression above. Moreover, $I_{n}(\gamma,\zeta)$ can be defined recursively for
  via
%   $\zeta\in[(n-1)\gamma^{-1},1]$, via
% \footnote{Leonardo: I removed the statement $\zeta\in[(n-1)\gamma^{-1},1]$. While this is
% true for the definition of the $H_n$, we need $\zeta\in(0,1)$ in the proofs. Also notice
% the addition of the term $\zeta\vee\gamma^{-1}$ in the lower limit of the integral defined
% by $I_1(\gamma,\zeta)$ in (23).}
%\adnote{There is no need for the $G$ functions!}
 \begin{equation}\label{e:recursionIn}
  I_{n}(\gamma,\zeta):=
   \begin{cases}
   \int_{\zeta\vee \gamma^{-1}}^1  L(y)\dd y,&\quad n=1,\\%[.25cm]
   \int_{\gamma^{-1}}^{\zeta-(n-2)\gamma^{-1}} L(y)\, I_{n-1}(\gamma,\zeta-y)\,\dd y,&\quad n\ge2.
   \end{cases}
  \end{equation}

%  In addition, for carrying our proofs we will require the following quantities which are defined recursively
% \begin{equation*}
%  G_{k}(\gamma,\zeta):=
%   \left\{\begin{array}{cc}
%     L(\zeta),&k=1,\\[.25cm]
%     \int_{\gamma^{-1}}^{\zeta-(k-1)\gamma^{-1}}L(y)\, G_{k-1}(\gamma,\zeta-y)\,\dd y& k\ge2.
%   \end{array}\right.
%  \end{equation*}
%
%
% \begin{lemma}\label{LemmaI2}
% \footnote{Leonardo: Removed the Lemma 3.  Just put the statement to be proved.}
Next we will prove that for $n=2,3,\ldots$, it holds that
 \begin{equation}\label{e:limsupterm}
%\adchange{
  \limsup_{\gamma\rightarrow\infty}
  \frac{I_n(\gamma,1)}{\gamma^{\alpha(n-2)}\log \gamma} = 0.
%}
 \end{equation}
% \end{lemma}
% \begin{proof}%[Proof of Lemma \ref{LemmaI2}]
% In the previous development we have demonstrated that
%  \begin{equation*}
%  \limsup_{\gamma\rightarrow\infty}\frac{\Em \II\{ M_d<\gamma-1\}\,Z^2}{\Ftail^{2-\epsilon}(\gamma)}
%  \le\limsup_{\gamma\rightarrow\infty}\frac{I_n(\gamma,1)}{\gamma^{\alpha(n-2)}\log \gamma}.
% \end{equation*}
% Hence it remains to prove that the limit of the right hand side is bounded.
%  We will exploit the recursive property %definition
%  of the $I_n(\gamma,1)$ functions.
%  The demonstration 
%  employs induction and is divided in several lemmas which are very
%  similar in spirit and therefore relegated to the appendix. % The proof below illustrates the main idea:
Since both numerator and denominator of \eqref{e:limsupterm} have limit $+\infty$,
 we can apply L'Hopital. 
% \footnote{Leonardo:  I am concerned about the use of L'Hopital
% with limsups.  Since we are considering the Pareto case we have limits everywhere.}. 
Lemma \ref{LemmaI} in the appendix provides a recursive expression for the derivative of the functions
 $I_n(\gamma,\zeta)$:
% \footnote{Leonardo: Lemma 6 was modified slightly changing 
%   $(n-1)$ to $n$ in the statement.} %in terms of the functions $I_k$ and $G_k$.
%\adchange{
 \begin{equation}\label{e:Inderivative}
 \frac{\partial}{\partial\gamma} I_n(\gamma,\zeta) =
 nL\big(\gamma^{-1}\big)\,I_{n-1}\big(\gamma,\zeta-\gamma^{-1}\big)\,\gamma^{-2},
 \qquad n=2,3,\ldots.
 \end{equation}
% with $\frac{\partial}{\partial\gamma} I_1(\gamma,\zeta)=0$.
 
 %to calculate the following limit in a recursive fashion:
%\begin{equation*}
% \lim_{u\rightarrow\infty}\frac{I_{k+1}(\gamma,1)}{\gamma^{\alpha(k-1)}\log \gamma}
%  \propto\left\{
%   \begin{array}{cc}
%	\lim\limits_{\gamma\rightarrow\inf	ty}\dfrac{\,L(\gamma^{-1})I_{1}(\gamma,1-\gamma^{-1})+I_1(\gamma,\gamma^{-1})\,G_{1}(\gamma,1-\gamma^{-1})}
%	{\gamma},&k=1\\[.25cm]
%    \lim\limits_{\gamma\rightarrow\infty}\dfrac{k\,L(\gamma^{-1})I_{k}(\gamma,1-\gamma^{-1})+I_1(\gamma,\gamma^{-1})\,G_{k}(\gamma,1-\gamma^{-1})}
%	{\gamma^{(k-1)\alpha+1}\log \gamma},&k\ge2.
%    \end{array}\right.
%\end{equation*}
%\adchange{
Therefore, we obtain
\begin{align}
& \limsup_{\gamma\to\infty} \frac{I_n(\gamma,1)}{\gamma^{\alpha(n-2)}\log \gamma}
= \limsup_{\gamma\to\infty}
\frac{\frac{\dd}{\dd\gamma}I_n(\gamma,1)}{\frac{\dd}{\dd\gamma}\gamma^{\alpha(n-2)}\log \gamma}\nonumber\\
&\quad = \limsup_{\gamma\to\infty}\frac{n\,L\big(\gamma^{-1}\big)\,I_{n-1}\big(\gamma,1-\gamma^{-1}\big)\,\gamma^{-2}}%
{\big(1+\alpha(n-2)\log\gamma\big)\,\gamma^{\alpha(n-2)-1}},\qquad n=2,3,\ldots.\label{e:hopital}
\end{align}
\begin{itemize}
\item
$n=2$. The expression in \eqref{e:hopital} becomes
\[
\frac{2\,L\big(\gamma^{-1}\big)\,I_{1}\big(\gamma,1-\gamma^{-1}\big)\,\gamma^{-2}}%
{\gamma^{-1}}
= \frac{2\,L\big(\gamma^{-1}\big)\,I_{1}\big(\gamma,1-\gamma^{-1}\big)}{\gamma}.
\]
Observe that
\begin{align*}
L\big(\gamma^{-1}\big)&=\big(1-\gamma^{-1}\big)^\alpha\,\gamma^{\alpha+1}
=\bigO\big(\gamma^{\alpha+1}\big),\qquad\gamma\to\infty;\\
L\big(1-\gamma^{-1}\big)&=\gamma^{-\alpha}\,\big(1-\gamma^{-1}\big)^{-(\alpha+1)}
=\bigO\big(\gamma^{-\alpha}\big),\qquad\gamma\to\infty;\\
I_{1}\big(\gamma,1-\gamma^{-1}\big)&=
\int_{1-\gamma^{-1}}^1 L(y)\,\dd y \leq \gamma^{-1} L\big(1-\gamma^{-1}\big)
=\bigO\big(\gamma^{-(\alpha+1)}\big),\qquad\gamma\to\infty,
\end{align*}
where the inequality follows  because the
function $L(y)$ is decreasing on $(0,1]$.
Hence,
\[
\limsup_{\gamma\to\infty}
\frac{2\,L\big(\gamma^{-1}\big)\,I_{1}\big(\gamma,1-\gamma^{-1}\big)}{\gamma}
=\limsup_{\gamma\to\infty}
\frac{(\mbox{a constant})\,\times\,\gamma^{\alpha+1}\,\gamma^{-(\alpha+1)}}{\gamma}=0.
\]
%Now observe that  $L(\gamma^{-1})\gamma^{-(\alpha+1)}\rightarrow1$
%while  ${G_1(\gamma,1-\gamma^{-1})}{\gamma^{\alpha}}={L(1-\gamma^{-1})}{\gamma^{\alpha}}\rightarrow1$
%as $\gamma\rightarrow\infty$.
\item
$n\ge 2$.
Assume \eqref{e:limsupterm} holds for $n$. Then reasoning as above
and using Lemma \ref{LemmaIG2} for equality (i), we get for $n+1$
\begin{align*}
&  \limsup_{\gamma\rightarrow\infty}
  \frac{I_{n+1}(\gamma,1)}{\gamma^{\alpha(n-1)}\log \gamma}\\
&\quad=\limsup_{\gamma\to\infty}
\frac{\frac{\dd}{\dd\gamma}I_{n+1}(\gamma,1)}%
{\frac{\dd}{\dd\gamma}\gamma^{\alpha(n-1)}\log \gamma}\\
&\quad= \limsup_{\gamma\to\infty}
\frac{(n+1)L\big(\gamma^{-1}\big)\,I_{n}\big(\gamma,1-\gamma^{-1}\big)\,\gamma^{-2}}%
{\big(1+\alpha(n-1)\log\gamma\big)\,\gamma^{\alpha(n-1)-1}}\\
&\quad\overset{\rm(i)}{=} \limsup_{\gamma\to\infty}
\frac{(n+1)L\big(\gamma^{-1}\big)\,\Big(I_{n}\big(\gamma,1\big)+o(1)\Big)\,\gamma^{-2}}%
{\big(1+\alpha(n-1)\log\gamma\big)\,\gamma^{\alpha(n-1)-1}}\\
&\quad=\limsup_{\gamma\to\infty}
\frac{(\mbox{a constant})\,\times\,\gamma^{\alpha+1} \,I_{n}\big(\gamma,1\big)\,\gamma^{-2} + o(1)}%
{(\mbox{a constant})\,\times\,\gamma^{\alpha(n-1)-1}\,\log\gamma}\\
&\quad=\limsup_{\gamma\to\infty}\ \ (\mbox{a constant})\,\times\,
\frac{I_{n}\big(\gamma,1\big) + o(1)}%
{\gamma^{\alpha(n-2)}\,\log\gamma}=0
\end{align*}
\end{itemize}

\par\bigskip\noindent
%\adchange{
% \textit{Proof of Lemma \ref{lem:secondtermpareto}}.
Putting together these arguments we can complete the proof of the Proposition:
\begin{align*}
& \limsup_{\gamma\to\infty}
\frac{\Em \II\{ M_d\le\gamma\}\,Z^2}{\ell^{2-\epsilon}}\\
&\quad\overset{\eqref{e:enough}}{\le}\limsup_{\gamma\to\infty}
\frac{1}{\ell^{2-\epsilon}}\;
\Em_f\II\{M_d\le \gamma, S_d>\gamma\}\prod_{i=1}^d \frac{\Ftail(\gamma)}{\Ftail(\gamma-X_i)}\\
&\quad\overset{\eqref{e:sumhn}}{\le}\limsup_{\gamma\to\infty}
\frac{\sum_{n=2}^d H_n(\gamma)}{\ell^{2-\epsilon}}\\
&\quad\overset{\eqref{e:defIn}}{=}\limsup_{\gamma\to\infty}
\sum_{n=2}^d \frac{\alpha^n  I_n(\gamma)}{\gamma^{\alpha n}\,\ell^{2-\epsilon}}
\end{align*}
Now notice that
\[
\ell = \Fdtail(\gamma)\ge \Ftail(\gamma)=\gamma^{-\alpha},
\]
thus, for $\epsilon<1/\alpha$ (that is, $\epsilon\alpha<1$)
\[
\ell^{2-\epsilon}\ge \gamma^{-2\alpha}\,\gamma^{\alpha\epsilon}
\ge \gamma^{-2\alpha}\,\log\gamma,\qquad\gamma\to\infty.
\]
Combining this with above, we get
\[
\limsup_{\gamma\to\infty}
\sum_{n=2}^d \frac{\alpha^n I_n(\gamma)}{ \gamma^{\alpha n}\,\ell^{2-\epsilon}}
\le
\sum_{n=2}^d \alpha^{n} \limsup_{\gamma\to\infty}
\frac{I_n(\gamma)}{\gamma^{\alpha (n-2)}\,\log\gamma}
=0
\]
\end{proof}
%}
%
%
%\renewcommand{\Q}{\mathrm{Q}}
%\subsubsection{Sum of Log-normal random variables}
%Without loss of generality assume that the pdf of the lognormal density  is $\phi(\log(x))/x$, where
%$\phi(x)=\frac{1}{\sqrt{2\pi}}\;\e^{-x^2/2}$ and the corresponding tail probability is $\Pm(X_i> x)=\Q(\log x)$.
%
%
%Then, the residual term in \eqref{decomp} is given by:
%\[
%\begin{split}
%\Em_f  \II\{M_d<\gamma ,S>\gamma\}\prod_{i=1}^d \frac{\Ftail^{(d)}(\gamma)}{\Ftail^{(d-1)}(\gamma-X_i)}&\leq \Em_f  \II\{M_d<\gamma ,S>\gamma\}\prod_{i=1}^d \frac{\Ftail^{(d)}(\gamma)}{\Ftail(\gamma-X_i)}\\
%&\leq \idotsint_{M_d<\gamma ,S>\gamma} \prod_{i=1}^d   \frac{\Q(\log\gamma)\phi(\log x_i)}{x_i \Q(\log(\gamma-x_i))} \di \bx\\
%&= \idotsint_{\mathscr{D}} \prod_{i=1}^d  \frac{\Q(\log\gamma)\phi(\log\gamma+\log u_i)}{u_i \Q(\log\gamma+\log(1-u_i))}    \di\bu
%\end{split}
%\]

\subsection{Light-tailed case}
\label{LightSection}
In this section we consider the case where $F$ belongs to a subfamily of light-tailed distributions
as defined by Embrechts and Goldie \cite{embrechts1980closure}. We say that a distribution $F$
belongs to the \emph{Embrechts-Goldie} family of distributions indexed by the parameter
$\theta\ge0$ and denoted $\mathcal{L}(\theta)$,  if
  \begin{equation}\label{ExpClass}
   \lim_{\gamma\rightarrow\infty}\frac{\Ftail(\gamma+x)}{\Ftail(\gamma)}=\e^{-\theta x}.
  \end{equation}
 If $\theta$ is strictly larger than $0$ then $\mathcal{L}(\theta)$ contains light-tailed
 distributions exclusively and is often referred
 as the \emph{exponential class}. This is a very rich class of distributions that
 includes several well know light-tailed distributions such as
  the exponential, gamma and phase-type.
 In contrast, if $\theta=0$, then $\mathcal{L}(0)$ corresponds to the class of \emph{long-tailed distributions}
 which is a large subclass of heavy-tailed distributions.
 In this section we concentrate on the light-tailed case $\theta>0$, but in order to derive our efficiency
 statements we draw some results for the class of the so called \emph{long-tailed functions} (cf. \cite
 [Definition 2.14]{foss2011introduction}).   More precisely, $h$ is long-tailed if it is
 ultimately positive and
\begin{equation}\label{LongTail}
 \lim_{\gamma\rightarrow\infty}\frac{h(\gamma+x)}{h(\gamma)}=1,\qquad \forall x.
\end{equation}
Obviously, if $F\in \mathcal{L}(0)$, then the tail probability $\Ftail$
is long tailed.
Important properties for the exponential class ($\theta> 0$) are
\begin{itemize}
\item
$\mathcal{L}(\theta)$  is closed under convolutions
\cite[Theorem 3]{embrechts1980closure}.
That is, if $F\in \mathcal{L}(\theta)$, then
   the $d$-fold convolution $F^{*d}\in\mathcal{L}(\theta)$.
\item
Define for $\alpha>0$ the distribution $G(x)=1-\big(\Ftail(x)\big)^\alpha$.
One can easily check that $G\in \mathcal{L}(\alpha\theta)$ whenever
  $F\in\mathcal{L}(\theta)$.

\item
The tail probability can be decomposed into the product of an exponential
   and a long tailed function
  \begin{equation}
    \label{LightDecomposition}
    \Ftail(\gamma)=\e^{-\theta \gamma}h(\gamma).
  \end{equation}
\end{itemize}
Decomposition \eqref{LightDecomposition} will be useful for proving efficiency
of the proposed estimator, but it is also interesting
on its own.
To verify it we define $h(\gamma):=\Ftail(\gamma)\;\e^{\theta\gamma}$.  Since $F\in\mathcal{L}(\theta)$ it follows that
  \begin{equation*}
   \lim_{\gamma\rightarrow\infty}\frac{h(\gamma+x)}{h(\gamma)}
   =\lim_{\gamma\rightarrow\infty}\frac{h(\gamma+x)\,\e^{-\theta(\gamma+x)}}{h(\gamma)\,\e^{-\theta(\gamma+x)}}
   =\lim_{\gamma\rightarrow\infty}\frac{\Ftail(\gamma+x)}{\Ftail(\gamma)\,\e^{-\theta x}}
   =1.
  \end{equation*}
The next property states that the asymptotic decay of a long-tailed function is slower
than the exponential rate \cite[Lemma 2.17]{foss2011introduction}.  More precisely, if $h$ is long tailed, then
\begin{equation}\label{DecayLong}
  \lim_{\gamma\rightarrow\infty}\frac{h(\gamma)}{\e^{-\epsilon \gamma}}=\infty,\qquad \forall\epsilon>0.
\end{equation}
These properties will be employed
to construct an asymptotic upper bound for the semi-parametric estimator.  In particular,
the following Lemma shows that the ratio of two tail convolutions of the same distribution in
$\mathcal{L}(\theta)$ cannot increase/decrease faster than at exponential rate.
\begin{lemma}\label{lem4}
	Let $F\in\mathcal{L}(\theta)$, $\theta>0$, and $d_1,d_2\in\nat$. Then
	${\overline{F^{*d_1}}(\gamma)}\Big/\,{\overline{F^{*d_2}}(\gamma)}=o(\e^{\epsilon \gamma}),\qquad \forall \epsilon>0$.
  \end{lemma}
  \begin{proof}
   Since $\mathcal{L}(\theta)$ is closed by convolution, then
   $F^{*d_1}, F^{*d_2}\in \mathcal{L}(\theta)$ and their tail distributions have decompositions
   as in \eqref{LightDecomposition} for some long tailed functions $h_1$ and $h_2$.  Therefore
  \begin{equation*}
   \frac{\overline{F^{*d_1}}(\gamma)}{\overline{F^{*d_2}}(\gamma)}=
   \frac{h_1(\gamma)\e^{-\theta \gamma}}{h_2(\gamma)\e^{-\theta \gamma}}=\frac{h_1(\gamma)}{h_2(\gamma)}.
  \end{equation*}
%\adchange{
  We first argue that both $h_1(\cdot)/h_2(\cdot)$ and its reciprocal function are long-tailed.
  This is so, because they are ultimately positive, and
  \[
  \frac{h_1(\gamma+x)/h_2(\gamma+x)}{h_1(\gamma)/h_2(\gamma)}
  =\frac{h_1(\gamma+x)}{h_1(\gamma)}\,\times\,\frac{h_2(\gamma)}{h_2(\gamma+x)}\to 1.
  \]
  The reciprocal function goes similarly.
  Thus, $h_2(\cdot)/h_1(\cdot)$
%}
  %Now, $h_1(\cdot)/h_2(\cdot)$ and its reciprocal function are strictly positive and satisfy
  %condition \eqref{LongTail}.  Hence both are long-tailed and
  satisfies condition \eqref{DecayLong}, which says
 \begin{equation*}
  \lim_{\gamma\rightarrow\infty}\frac{h_2(\gamma)/h_1(\gamma)}{\e^{-\epsilon\gamma}}=\infty.
 \end{equation*}
  Clearly, this is equivalent to
 \begin{equation*}
  \lim_{\gamma\rightarrow\infty}\frac{h_1(\gamma)/h_2(\gamma)}{\e^{\epsilon\gamma}}=0.
 \end{equation*}
% This completes the proof of the Lemma.
\halmos
  \end{proof}
We also have the following.

\noindent   \textbf{Assumption A:}
   \label{lem5}
   Let  $h$ be a long-tailed function such that $h(x)>0$ for all $x\ge0$.  Then
   $G(\gamma):=\sup\{h(\gamma)/h(x):0\le x\le \gamma\}=o(\e^{\epsilon\gamma})$ for all $\epsilon>0$.
%   \end{lemma}
  %\begin{proof}
%   The proof is lengthy, but straightforward, and thus omitted. 
   %Let $\mathcal{I}:=\inf\{h(y):y\ge0\}$. If $\mathcal{I}>0$, then $G(\gamma)\le h(\gamma)/I$ which by Lemma
%   \ref{lem3} is an $\oh(\e^{\epsilon\gamma})$ function for all $\epsilon>0$.  Otherwise, if
%   $I=0$ then  FINISH THIS LONG PROOF...
  %\end{proof}

\begin{proposition}[Logarithmic efficiency of $\hat\ell$]
If Assumption A holds, the estimator $Z=\Ind\{S(\bX)>\gamma\}\frac{f(\bX)}{g(\bX)}$ satisfies
\[
\lim_{\gamma\uparrow\infty}\frac{\Em Z^2}{\ell^{2-\epsilon}(\gamma)}=0,\qquad \forall \epsilon>0\;.
\]
\end{proposition}

\begin{proof}
%\adchange{
Recall
\[
\Em Z^2 = \Em_f\II\{S(\bX)>\gamma\}\,
\prod_{i=1}^d \frac{\Fdtail(\gamma)}{\Fdmtail(\gamma-X_i)}.
\]
%}
 We write
  \begin{equation*}
 %  \Exp\bigg[\prod_{i=1}^d \frac{\overline{F^{*d}}(u)}{\overline{F^{*d-1}}(u-X_i)}\Ind_{\{S_d>u\}}\bigg]
   \prod_{i=1}^d\frac{\Fdtail(\gamma)}{\Fdmtail(\gamma-X_i)}
      =   H(\gamma)\prod_{i=1}^d
  \frac{\Fdmtail(\gamma)}{\Fdmtail(\gamma-X_i)},
  \end{equation*}
 where $H(\gamma):= \Big[\Fdtail(\gamma)\big/\Fdmtail(\gamma)\Big]^d$.
 Since $F^{*(d-1)}\in\mathcal{L}(\theta)$ we can
 use the decomposition \eqref{LightDecomposition} to write
 $\Fdmtail(\gamma)=h(\gamma)\e^{-\theta\gamma}$ for some $h(\cdot)$
 long tailed function.  Hence, we obtain the following bound
 \begin{equation*}
  \prod_{i=1}^d\frac{\Fdmtail(\gamma)}{\Fdmtail(\gamma-X_i)}
  =\prod_{i=1}^d\frac{h(\gamma)}{h(\gamma-X_i)}\;\frac{\e^{-\theta \gamma}}{\e^{-\theta(\gamma-X_i)}}
  \le\bigg(\sup_{0\le x\le \gamma}\frac{h(\gamma)}{h(\gamma-x)}\bigg)^d\prod_{i=1}^d{\e^{-\theta X_i}}
  =\big(G(\gamma)\big)^d\,\e^{-\theta S(\bX)}
 \end{equation*}
 where $G(\gamma):=\sup_{0\le x\le \gamma}\big\{{h(\gamma)}\big/{h(\gamma-x)}\big\}$.
 Using these we obtain
   \[
  %\limsup_{\gamma\rightarrow\infty}
  \frac{\Exp Z^2}{\ell^{2-\epsilon}(\gamma)}
  %=\limsup_{\gamma\rightarrow\infty}\frac{1}{\ell^{2-\epsilon}(\gamma)}\;
	%\int\limits_{S_d>\gamma}\prod_{i=1}^d
  %\frac{\overline{F^{*d}}(\gamma)}{\overline{F^{*d-1}}(\gamma-X_i)} d\Prob
  \le%\limsup_{\gamma\rightarrow\infty}
  \frac{H(\gamma)G^d(\gamma)}{\ell^{2-\epsilon}(\gamma)}%\int\limits_{S_d>\gamma}
   \Em_f \II\{S(\bX)>\gamma\}\, \e^{-\theta S(\bX)},
  \]
%\adchange{
  where $\theta>0$. Hence,
  \[
  \Em_f \II\{S(\bX)>\gamma\}\, \e^{-\theta S(\bX)}\leq \e^{-\theta\gamma}\Pm_f(S(\bX)>\gamma)
  =\e^{-\theta\gamma}\,\ell.
  \]
  Thus we get
  \[
  \frac{\Exp Z^2}{\ell^{2-\epsilon}(\gamma)}\leq
  \frac{H(\gamma)G^d(\gamma)\e^{-\theta \gamma}}{\ell^{1-\epsilon}(\gamma)}.
  \]
%}
 %In the set $\{S_d>\gamma\}$, the function $\e^{-\theta S_d}$ is bounded above by $\e^{-\theta\gamma}$.
 %Also recall that $\ell=\Ftail(\gamma)$, so
 %\begin{align*}
 %\limsup_{\gamma\rightarrow\infty}\frac{H(\gamma)G^d(\gamma)}{\ell^{2-\epsilon}(\gamma)}\int _{S_d>\gamma}
 % {\e^{-\theta S_d}}{}d\Prob
 % &\le\limsup_{\gamma\rightarrow\infty}\frac{H(\gamma)G^d(\gamma)\e^{-\theta \gamma}}{\ell^{2-\epsilon}(\gamma)}
 % \int _{S_d>\gamma}d\Prob=\limsup_{\gamma\rightarrow\infty}\frac{H(\gamma)G^d(\gamma)\e^{-\theta \gamma}}{\ell^{1-\epsilon}(\gamma)}
 %\end{align*}
%\adchange{
 Applying the properties of the exponential class
%}
 %Notice that $\ell^{1-\epsilon}=[F^{*d}]^{1-\epsilon}\in\mathcal{L}(\theta(1-\epsilon))$ so
we can write
\[
\ell^{1-\epsilon}=\big(\Fdtail(\gamma)\big)^{1-\epsilon}
=\e^{-\theta(1-\epsilon)\gamma}h_d(\gamma)
\]
 for some long tailed function $h_d$. In consequence, %the last limit is equivalent to
 \begin{align*}
  \limsup_{\gamma\rightarrow\infty} &  \frac{\Exp Z^2}{\ell^{2-\epsilon}(\gamma)}
  \leq \limsup_{\gamma\rightarrow\infty}
  \frac{H(\gamma)G^d(\gamma)\e^{-\theta \gamma}}{\ell^{1-\epsilon}(\gamma)}\\
  &= \limsup_{\gamma\rightarrow\infty}
  \frac{H(\gamma)G^d(\gamma)\e^{-\theta \gamma}}{h_d(\gamma)\e^{-(1-\epsilon)\theta\gamma}}=
  \limsup_{\gamma\rightarrow\infty}\frac{H(\gamma)G^d(\gamma)}{h_d(\gamma)}\e^{-\epsilon\theta \gamma}.
 \end{align*}
 Now, property \eqref{LightDecomposition} and Lemma \ref{lem4} and Lemma \ref{lem5}
 imply that none of the functions $H$, $G$, $h_d^{-1}$ and their product cannot increase
 at exponential rate, namely $H(\gamma)G^d(\gamma)/h_d(\gamma)=o(\e^{\theta\epsilon \gamma})$.
 Hence, the last limit is 0.
 %This completes the proof.
\halmos
\end{proof}

\section{Conclusions}
\label{sec:conclusion}

In this paper we have described a procedure for implementing  an optimal cross-entropy importance sampling density for the purpose of estimating a rare-event probability, indexed by the rarity parameter $\gamma$. The goal is to estimate the optimal importance sampling density for a finite $\gamma$ within the class of all densities in product form.  This optimal importance sampling density is typically not available analytically and this is why in  practical simulations we estimate it via MCMC simulation from the minimum variance pdf. The numerical examples suggest that the resulting estimator can yield significantly better efficiency  compared to many currently recommended estimators. The same procedure is efficient in both light- and heavy-tailed cases.  This is especially relevant for probabilities involving the Weibull distribution with tail index $\alpha<1$, but close to unity. This setting yields behavior  intermediate between the typical heavy- and light-tailed behavior expected of rare-events. As a result,  while existing procedures are inefficient or fail completely, our method  estimates reliably Weibull probabilities for any values of $\alpha$, including $\alpha>1$.

The  practical implementation of the proposed method depends on a preliminary MCMC step, which  is  a powerful, but poorly understood heuristic  that needs further investigation.
% because it is essential for the practical implementation of the algorithm and for achieving the good performance we have observed.   
In this article we have established the efficiency of the method in the light- and heavy-tailed case, but have done so by ignoring any errors arising from the preliminary MCMC step.  Future work will need to address the impact of the  MCMC approximation on the quality of the estimator. A good starting point for such an analysis might be to consider  the probabilistic relative error  efficiency concept introduced in
 \cite{tuffin2012probabilistic}.
%
%This is for another paper:
%
%From properties of MCMC and strong law of large numbers, we have the almost sure uniform convergence \cite{gelfand1990sampling,pratt1960interchanging,schervish1992convergence}:
%\[
%\sup_{\bx}\frac{\hat\pi(\bx)}{\pi(\bx)} \stackrel{a.s.}{\rightarrow} 1
%\]
%
%From here we have bounded strong probabilistic relative error properties \cite{tuffin2012probabilistic}.
%
%\[
%\int_\mathbb{R} |\pi_i^{[n]}(x)-\pi_i(x)|\di x\leq  \int_{\mathbb{R}^d} |\pi^{[n]}(\bx)-\pi(\bx)|\di \bx
%\]
%\[
%\frac{1}{2}\int|\prod \pi_i^{[n]}-\prod \pi_i |\leq 1-\prod_i\left(1-\frac{1}{2}\int|\pi_i^{[n]}-\pi_i|\right)
%\]
%
%
%Use the fact that total variation norm is:
%\[ \sup_A\left|\int f-\int p\right|=\frac{1}{2}\int|f-p|+\frac{1}{2}\left|\int f -\int p\right| \]
%\[
%\begin{split}
%\sup_A |\hat\pi(A)-\pi(A)|&=\frac{1}{2}\int |\hat\pi(\by)-\pi(\by)|\di \by\\
%&=1-\Em_{\hat\pi}\min\left\{\frac{\pi(\bY)}{\hat\pi(\bY)},\;1\right\}=1-\int\min\{\pi,\hat\pi\}\\
%&=\int(\pi-\hat\pi)^+\\
%&= \frac{1}{2}\sup_{|\phi|\leq 1}\Big|\int \phi(\by)\hat\pi(\by)\di\by-\int\phi(\by)\pi(\by)\di\by\Big|\\
%&=\sup_A (\hat\pi(A)-\pi(A))-\inf_A (\hat\pi(A)-\pi(A))
%\end{split}
%\]

{\small

\section{Appendix}
\label{sec:appendix}

\subsection{Proofs. Section \ref{subsec:semiparm_IS}}

\begin{proof}[Proof of Lemma \ref{Lemma1}]
First note that for any single-variate function $h$:
\begin{align*}
\int_{\R^d} & h(x_1)\pi(\bx)\,\di\bx= \int_{\R} h(x_1)
\Big( \int_{\R^{d-1}}\pi(x_1,x_2,\ldots,x_d)\,\di x_2\cdots\di x_d \Big)\,\di x_1\\
&= \int h(x_1)\pi_1(x_1)\,\di x_1.
\end{align*}
Next, using the properties of the cross-entropy distance we have that
\[
\pi_1 =\argmin_{g_1\in\scG_1} \int \pi_1(x_1) \log\left(\frac{\pi_1(x_1)}{g_1(x_1)}\right)\di x_1=\argmax_{g_1\in\scG_1}\int \pi_1(x_1)\log g_1(x_1)\,\di x_1.
\]
Applying these two observations for any $i=1,\ldots,d$ gives
\begin{align*}
\argmax_{g_1,\ldots,g_d\in\scG_1} &
\int\pi(\bx)\log\left(\prod_{i=1}^d g_i(x_i)\right)\di \bx\\
&= \argmax_{g_1,\ldots,g_d\in\scG_1}
\sum_{i=1}^d \int\pi(\bx)\log g_i(x_i)\, \di \bx\\
&= \argmax_{g_1,\ldots,g_d\in\scG_1}
\sum_{i=1}^d\int \pi_i(x_i)\log g_i(x_i)\, \di x_i= \sum_{i=1}^d \argmax_{g_i\in\scG_1}\int \pi_i(x_i)\log g_i(x_i)\, \di x_i,
\end{align*}
from where we obtain the solution $g_i=\pi_i$ for all $i=1,\ldots,d$.
\halmos
\end{proof}

\subsection{Proofs.  Section \ref{HeavySection}}
%\adchange{
\begin{lemma}\label{LemmaI}
Assume $\zeta\ge n\gamma^{-1}$.  Then
%  \begin{align*}
%  \frac{\dd}{\dd \gamma}I_{k+1}(\gamma,\zeta)
%  &=\frac{k\,L(\gamma^{-1})}{\gamma^2}I_{k}(\gamma,\zeta-\gamma^{-1})
%	+\frac{I_{1}(\gamma,\gamma^{-1})}{\gamma^2} G_{k}(\gamma,\zeta-\gamma^{-1}),\qquad k\ge1.
% \end{align*}
\begin{equation}\label{e:derIn}
 \frac{\partial}{\partial\gamma} I_n(\gamma,\zeta) =
 n\,L\big(\gamma^{-1}\big)\,I_{n-1}\big(\gamma,\zeta-\gamma^{-1}\big)\,\gamma^{-2},
 \qquad n=2,3,\ldots.
\end{equation}
\end{lemma}
\begin{proof}%[Proof of Lemma \ref{LemmaI}]
The proof is by induction. Recall the recursive introduction of the $I_n$
functions:
\begin{align*}
I_{1}(\gamma,\zeta)&= \int_{\zeta\,\vee\,\gamma^{-1}}^1  L(y)\,\dd y;\\
I_{n}(\gamma,\zeta)&= \int_{\gamma^{-1}}^{\zeta-(n-2)\gamma^{-1}} L(y)\,
I_{n-1}(\gamma,\zeta-y)\,\dd y,\qquad n=2,3,\ldots
\end{align*}
First consider 
\begin{align*}
& \frac{\partial}{\partial\gamma}I_{2}(\gamma,\zeta)
=\frac{\partial}{\partial\gamma}\int_{\gamma^{-1}}^{\zeta-\gamma^{-1}}L(y)\,
I_{1}(\gamma,\zeta-y)\,\dd y+
\frac{\partial}{\partial\gamma}\,\int_{\zeta-\gamma^{-1}}^{\zeta}L(y)\,\dd y\,
I_{1}(\gamma,\gamma^{-1})\\
&\quad=\bigg[
L(\zeta-\gamma^{-1})I_1(\gamma,\gamma^{-1})
-L\big(\gamma^{-1}\big)\,I_{1}\big(\gamma,\zeta-\gamma^{-1}\big)
-L(\zeta-\gamma^{-1})I_1(\gamma,\gamma^{-1})
-I_1(\gamma,\zeta-\gamma^{-1})L(\gamma^{-1})
\bigg]\, \frac{\dd}{\dd \gamma}\gamma^{-1}\\
&\quad=2\,L\big(\gamma^{-1}\big)\,I_{1}\big(\gamma,\zeta-\gamma^{-1}\big)\,\gamma^{-2}.
\end{align*}
Next, assume that \eqref{e:derIn} holds for $n$. Then
\begin{align*}
& \frac{\partial}{\partial\gamma}I_{n+1}(\gamma,\zeta)
=\frac{\partial}{\partial\gamma} \int_{\gamma^{-1}}^{\zeta-(n-1)\gamma^{-1}} L(y)\,
I_{n}(\gamma,\zeta-y)\,\dd y\\
&\quad=L(\zeta-(n-1)\gamma^{-1})\,I_{n}\big(\gamma,(n-1)\gamma^{-1}\big)\,\frac{\dd}{\dd \gamma}\big(\zeta-(n-1)\gamma^{-1}\big)
-L(\gamma^{-1})\,I_{n}\big(\gamma,\zeta-\gamma^{-1}\big)\,\frac{\dd}{\dd \gamma}\gamma^{-1}\\
&\qquad + \int_{\gamma^{-1}}^{\zeta-(n-1)\gamma^{-1}} L(y)\,
\frac{\partial}{\partial\gamma} I_{n}(\gamma,\zeta-y)\,\dd y\\
&\quad=0 + L(\gamma^{-1})\,I_{n}\big(\gamma,\zeta-\gamma^{-1}\big)\,\gamma^{-2}
+ \int_{\gamma^{-1}}^{\zeta-(n-1)\gamma^{-1}} L(y)\,
n\,L\big(\gamma^{-1}\big)\,I_{n-1}\big(\gamma,\zeta-\gamma^{-1}-y\big)\,\gamma^{-2}\,\dd y\\
&\quad=L(\gamma^{-1})\,I_{n}\big(\gamma,\zeta-\gamma^{-1}\big)\,\gamma^{-2}
+ n\,L\big(\gamma^{-1}\big)
\int_{\gamma^{-1}}^{\zeta-\gamma^{-1}-(n-2)\gamma^{-1}} L(y)\,I_{n-1}\big(\gamma,\zeta-\gamma^{-1}-y\big)\,\dd y \,\gamma^{-2}\\
&\quad=L(\gamma^{-1})\,I_{n}\big(\gamma,\zeta-\gamma^{-1}\big)\,\gamma^{-2}
+ n\,L\big(\gamma^{-1}\big)I_n\big(\gamma,\zeta-\gamma^{-1}\big)\,\gamma^{-2}\\
&\quad=(n+1)L(\gamma^{-1})\,I_{n}\big(\gamma,\zeta-\gamma^{-1}\big)\,\gamma^{-2}
\end{align*}
\halmos
\end{proof}
\begin{lemma}\label{LemmaIG2}
For $n=1,2,\ldots$:
 \begin{equation}\label{e:asymptoticIG}
 I_{n}\big(\gamma,\zeta-\gamma^{-1}\big)=
 I_{n}(\gamma,\zeta)+o(1),\qquad \gamma\to\infty.
\end{equation}
\end{lemma}
\begin{proof}
Apply induction and the
the recursive definition of $I_n$ functions.
\begin{itemize}
\item
$n=1$.
\begin{align*}
I_{1}&\big(\gamma,\zeta-\gamma^{-1}\big)=
\int_{\zeta-\gamma^{-1}}^1 L(y)\,\dd y\\
&= I_1(\gamma,\zeta) + \int_{\zeta-\gamma^{-1}}^{\zeta} L(y)\,\dd y\\
&= I_1(\gamma,\zeta) + \gamma^{-1} L(\eta),
\end{align*}
for some $\eta\in(\zeta-\gamma^{-1},\zeta)$
(mean value theorem). Clearly, the second term
is $o(1)$ for $\gamma\to\infty$.
\item
$n\ge 1$. Assume \eqref{e:asymptoticIG} holds. Then
\begin{align*}
I_{n+1}&\big(\gamma,\zeta-\gamma^{-1}\big)=
\int_{\gamma^{-1}}^{\zeta-n\gamma^{-1}}L(y)I_{n}\big(\gamma,\zeta-\gamma^{-1}-y\big)\,\dd y\\
&= \int_{\gamma^{-1}}^{\zeta-(n-1)\gamma^{-1}}L(y)\Big(I_{n}(\gamma,\zeta-y)+o(1)\Big)\,\dd y
-\int_{\zeta-n\gamma^{-1}}^{\zeta-(n-1)\gamma^{-1}}L(y)I_{n}\big(\gamma,\zeta-\gamma^{-1}-y\big)\,\dd y\\
&=I_{n+1}(\gamma,\zeta) + o(1)\,\int_{\gamma^{-1}}^{\zeta-(n-1)\gamma^{-1}}L(y)\,\dd y
-\gamma^{-1}L(\eta)I_{n}\big(\gamma,\zeta-\gamma^{-1}-\eta\big)\\
&= I_{n+1}(\gamma,\zeta) + o(1),\qquad\gamma\to\infty.
\end{align*}
\end{itemize}
\halmos
\end{proof}

}

\bibliographystyle{plain}
\bibliography{mcis}
\end{document}